\PassOptionsToPackage{x11names}{xcolor}
\documentclass[
  a4paper,
  UKenglish,
  cleveref,
  autoref,
  thm-restate
]{lipics-v2021}
\nolinenumbers{}
\usepackage[
]{apxproof}
\usepackage{xfp}
\usepackage{makecell}
\usepackage{amsmath,amssymb}
\usepackage{booktabs}
\usepackage{tabularray}
\usepackage{bm}

\newcounter{myappendix}
\renewcommand{\themyappendix}{\Alph{myappendix}}
\newcommand{\appendixsection}[1]{%
  \refstepcounter{myappendix}
  \section*{\colorbox{lipicsYellow}{\kern0.15em\themyappendix\kern0.15em}\quad #1}%
}

\usepackage{etoolbox}

\makeatletter
\g@addto@macro\appendixprelim{
  \renewcommand{\axp@section}[2]{}
}
\ifthenelse{\equal{\axp@appendix}{inline}}{
  \renewenvironment{toappendix}{%
    \par\addvspace{5ex}\noindent
    \begin{tikzpicture}
      \draw[lipicsLightGray, line width=1pt] (0,0) -- (\textwidth,0)
      node[midway, fill=white, inner sep=2pt, text=black] {\tiny \sffamily\color{lipicsGray}APPENDIX START};
    \end{tikzpicture}
    \par\addvspace{-2ex}
  }{%
    \par\addvspace{0ex}\noindent
    \begin{tikzpicture}
      \draw[lipicsLightGray, line width=1pt] (0,0) -- (\textwidth,0)
      node[midway, fill=white, inner sep=2pt, text=black] {\tiny  \sffamily\color{lipicsGray}APPENDIX FINISH};
    \end{tikzpicture}
    \par\addvspace{5ex}
  }
}{}
\makeatother

\usepackage{enumitem}
\usepackage{diagbox}
\usepackage{stmaryrd}
\usepackage{complexity}
\usepackage{textgreek}
\usepackage[noEnd=false,indLines=false]{algpseudocodex}
\usepackage{mleftright}
\mleftright{}
\usepackage{tikz}
\usepackage{varwidth}
\usetikzlibrary{arrows.meta,positioning, decorations.pathreplacing, shapes,calc}
\newenvironment{fakedisplaymath}{
  \par\nointerlineskip
  \vspace{\abovedisplayskip}
  \( \displaystyle
  }{
  \)
  \vspace{\belowdisplayskip}
}

\definecolor{mBlue}{HTML}{2E86C1}
\definecolor{mGreen}{HTML}{239B56}
\definecolor{mRose}{HTML}{CB4335}
\definecolor{mYellow}{HTML}{F39C12}
\definecolor{mDark}{HTML}{566573}
\definecolor{mGrey}{HTML}{95A5A6}

\let\OldW\W
\renewcommand{\W}[1]{\OldW\textup{\textsf{[#1]}}}

\newcommand{\tw}{\ensuremath\mathrm{tw}}
\newcommand{\pw}{\ensuremath\mathrm{pw}}
\newcommand{\ctw}{\ensuremath\mathrm{ctw}}

\newcommand{\pre}{N^{\preceq}} \newcommand{\fut}{N^{\succ}}
\newcommand{\pwseth}{$\pw$\textsc{-SETH}}
\newcommand{\gsym}{\textup{\sffamily co}}
\newcommand{\gname}{co-neighbor graph}
\renewcommand{\st}{\mathrm{St}}
\newcommand{\twcoe}{%
  \mathchoice
  {\textstyle \lfloor\frac{2\Delta}{3}+1\rfloor}
  {\lfloor\frac{2\Delta}{3}+1\rfloor}
  {\lfloor\frac{2\Delta}{3}+1\rfloor}
  {\lfloor\frac{2\Delta}{3}+1\rfloor}
}
\newcommand{\pwcoe}{%
  \mathchoice
  {\textstyle\lfloor\frac{\Delta}{2}+1\rfloor}
  {\lfloor\frac{\Delta}{2}+1\rfloor}
  {\lfloor\frac{\Delta}{2}+1\rfloor}
  {\lfloor\frac{\Delta}{2}+1\rfloor}
}

\tikzstyle{fill vertex}=[circle,fill=mDark,draw=none, minimum size=3pt,inner
sep=0pt]
\tikzstyle{vanish vertex}=[circle,fill=black,draw=none, minimum size=0.5pt,inner
sep=0pt]
\tikzstyle{edge}=[mGrey]
\tikzstyle{dashed edge}=[dashed, mGrey]
\tikzstyle{diedge}=[mGrey, ->, >={Stealth[round]}, line width=0.25pt]
\tikzstyle{dashed diedge}=[dashed, mGrey, ->, >={Stealth[round]}, line
width=0.25pt]

\newcommand{\drawHpoints}[4]{
  \foreach \i in {1,2,3} {
    \begin{scope}[shift={({(\i-1)*2}, {(\i-1)*0.2})}]
      \foreach \j/\mangle/\mrad/\mypage in { 1/90/1/1,2/90/2/1,3/90/3/1,
        1/210/1/2,2/210/2/2,3/210/3/2, 1/330/1/3,2/330/2/3,3/330/3/3%
      }{
        \pgfmathparse{((#1==-1 || #1==\i) && (#2==-1 || #2==\j) && (#3==-1 || #3==\mypage))}
        \edef\ismatch{\pgfmathresult}
        \pgfmathparse{\ismatch ? "#4" : "mDark"}
        \edef\vcolor{\pgfmathresult}
        \pgfmathparse{\ismatch ? "fill vertex" : "vanish vertex"}
        \edef\vstyle{\pgfmathresult}
        \node[\vstyle,fill=\vcolor] (v\mypage\j_\i) at (\mangle:\mrad) {};
      }

      \draw[edge] (v11_\i) -- (v21_\i) -- (v31_\i) -- (v11_\i);
      \draw[edge] (v11_\i) -- (v12_\i) -- (v13_\i);
      \draw[edge] (v21_\i) -- (v22_\i) -- (v23_\i);
      \draw[edge] (v31_\i) -- (v32_\i) -- (v33_\i);
    \end{scope}
  }

  \foreach \j [evaluate=\j as \nextj using int(\j+1)] in {1,2} {
    \foreach \mypage in {1,2,3} {
      \foreach \k in {1,2,3} {
        \draw[dashed, edge] (v\mypage\k_\j) -- (v\mypage\k_\nextj);
      }%
    }%
  }%
}

\newcommand{\drawH}{ \foreach \i in {1,2} {
    \begin{scope}[shift={({(\i-1)*2*1.5}, {(\i-1)*0.2*1.5})}]
      \foreach \j/\mangle/\mrad/\mypage in { 1/90/1/1,3/90/3/1,
        1/210/1/2,3/210/3/2, 1/330/1/3,3/330/3/3%
      } {%
        \coordinate (v\mypage\j_\i) at (\mangle:\mrad) {};
      }

      \draw[edge] (v11_\i) -- (v21_\i) -- (v31_\i) -- (v11_\i);
      \draw[edge] (v11_\i)  -- (v13_\i);
      \draw[edge] (v21_\i)  -- (v23_\i);
      \draw[edge] (v31_\i) -- (v33_\i);
    \end{scope}
  }

  \foreach \j [evaluate=\j as \nextj using int(\j+1)] in {1} {
    \foreach \mypage in {1,2,3} {
      \foreach \k in {1,3} {
        \draw[edge] (v\mypage\k_\j) -- (v\mypage\k_\nextj);
      }%
    }%
  }
}

\newcommand{\coverafter}{
  \foreach \i/\j/\h/\k in {v11_1/v11_2/v31_2/v31_1}{
    \draw[fill=white,draw=none,opacity=0.7] (\i) -- (\j) -- (\h) -- (\k) -- (\i);
    \draw[edge] (\i) -- (\j) -- (\h) -- (\k) -- (\i);
  } %
}

\title{Pure Nash Equilibria in Graphical Games of Bounded Width Revisited}
\author{Michael Lampis}{LAMSADE, CNRS UMR7243, Université Paris Dauphine-PSL, 75775 Paris, France}
{michail.lampis@dauphine.fr} {https://orcid.org/0000-0002-5791-0887}{}
\author{Yiren Lu}{LAMSADE, CNRS UMR7243, Université Paris Dauphine-PSL, 75775 Paris, France}
{yiren.lu@dauphine.psl.eu} {https://orcid.org/0009-0005-2846-4952}{}

\authorrunning{M. Lampis and Y. Lu}
\Copyright{Michael Lampis and Yiren Lu}

\EventEditors{Philip Bille, Seth Pettie, and Sabine Storandt}
\EventNoEds{3}
\EventLongTitle{34th Annual European Symposium on Algorithms (ESA 2026)}
\EventShortTitle{ESA 2026}
\EventAcronym{ESA}
\EventYear{2026}
\EventDate{August 31--September 4, 2026}
\EventLocation{L'Aquila, Italy}
\EventLogo{}
\SeriesVolume{388}
\ArticleNo{72}

\hideLIPIcs

\begin{document}
\maketitle

\begin{abstract}
  We revisit the complexity of deciding whether a graphical game admits a pure
  Nash equilibrium (PNE) parameterized by standard measures of the input graph,
  such as treewidth. The natural dynamic programming algorithm for this problem
  has parameter dependence $\alpha^{(\Delta+1)\tw}$ where $\alpha$ is the
  maximum number of strategies available to each player, each player's utility
  depends on at most $\Delta$ other players, and the input graph has width
  $\tw$. Our first contribution is to point out that an algorithm by Thomas and
  van Leeuwen [Algorithmica 2015] claiming to improve this dependence to
  $\alpha^{O(\tw)}$ is flawed and, more strongly, such an algorithm would imply
  that \FPT=\W{1}.

  We then set out to pinpoint the fine-grained complexity of this problem with
  respect to standard parameters and show that the natural DP algorithm is not
  optimal, as the problem can be solved with dependence $\alpha^{\twcoe{}\tw}$,
  $\alpha^{\pwcoe{}\pw}$, and $\alpha^{\ctw}$, where $\pw,\ctw$ are the
  pathwidth and cutwidth of the input respectively. Our main algorithmic tool is
  a tightening of the relationship between the width of a graph $G$, its maximum
  degree, and the width of $G^2$, which may be of independent interest.
  Complementing these results, we show that our algorithms for pathwidth and
  cutwidth are likely to be optimal, as improving them is equivalent to
  falsifying the \pwseth.
\end{abstract}

\section{Introduction}

The computation of Nash equilibria is a central topic at the intersection of
computer science and economics. In this paper we focus on equilibria in
\emph{graphical games}, that is, games involving $n$ players represented by the
vertices of a (di-)graph whose edges indicate player interactions, in the sense
that the presence of an arc $(u,v)$ indicates that the utility of player $u$
(partially) depends on the strategy of player $v$ (and the absence of an
arc indicates that $u$ is indifferent to $v$'s actions).  This is an extremely
natural and well-studied model~\cite{JiangS10,KearnsLS01,KollerM03}.

The question we are interested in is the complexity of deciding whether a given
graphical game admits a \emph{pure} Nash equilibrium (PNE), that is, a joint
strategy where no player can unilaterally increase her utility by changing her
strategy.  This problem is \NP-complete~\cite{SchoenebeckV12}, so we attack it
using the tools of parameterized complexity, which is one of the most
established approaches for dealing with NP-hardness\footnote{We assume the
  reader is familiar with the basics of FPT algorithms as given for example
in~\cite{CyganFKLMPPS15}.}. Our goal is to investigate the \emph{structural
parameterized complexity} of deciding if a graphical game admits a PNE, for
standard parameters such as treewidth, pathwidth, cutwidth and maximum degree.

This question is of course anything but new and in fact its study goes back more
than twenty years.  In particular, the works of Gottlob, Greco, and
Scarcello~\cite{GottlobGS05} and Daskalakis and
Papadimitriou~\cite{DaskalakisP06} established the following: if we are given a
graphical game $\mathcal{G}$ where each player has at most $\alpha$ available
strategies, each player's utility depends on at most $\Delta$ other players, and
we are supplied a tree decomposition of the game graph with width $\tw$, then a
PNE (if one exists) can be found in time
$\alpha^{(\Delta+1)\tw}|\mathcal{G}|^{O(1)}$. Even though the two works arrive
at this complexity through different paths (CSPs of bounded hypertreewidth
for~\cite{GottlobGS05} and Markov Random Fields for~\cite{DaskalakisP06}), the
heart of the two algorithms is essentially the natural dynamic programming
approach: for each bag we need to remember the strategies of the players in the
bag ($\alpha^{\tw}$) and their (out-)neighbors ($\alpha^{\Delta\tw}$).  Is this
algorithm optimal?

This important question was taken up by Thomas and van
Leeuwen~\cite{thomas_pure_2015}. Their main (purported) contribution was an
algorithm significantly improving upon the performance of the two previous
algorithms by removing the dependence on $\Delta$.  Specifically, they claim
that the problem can be decided in time $\alpha^{O(\tw)}|\mathcal{G}|^{O(1)}$.
As some dependence on $\Delta$ remains in the $|\mathcal{G}|^{O(1)}$
factor\footnote{Observe that if we have $n$ players, then the description of
  $\mathcal{G}$ has size $n\alpha^{\Delta+1}$, as for each player we are given a
  payoff matrix with her utility for each possible joint strategy of her closed
neighborhood.}, their algorithm attempts to separate the exponential dependence
on $\tw$ from the exponential dependence on $\Delta$. To the best of our
knowledge, this has so far been accepted as the current state of the art (for
instance~\cite{PapadimitriouP23}).  A main motivation of our paper is the
observation that this result is, in fact, \emph{flawed}, and indeed no such
algorithm is possible under standard hypotheses.

\subparagraph*{Our results} Our goal in this paper is two-fold: to clarify the
state of the art with respect to the structural
parameterized complexity of the fundamental problem of computing PNEs in
graphical games; and to give a fine-grained investigation of the precise
complexity of this problem with respect to some of the most important
parameters.

With respect to the first goal, our main result is that the algorithm
of~\cite{thomas_pure_2015} contains a serious flaw. More strongly, we show that
even when $\alpha=2$, the problem is \W{1}-hard parameterized by treewidth (or
even by vertex cover), which implies that if the algorithm
of~\cite{thomas_pure_2015} were correct, then \FPT=\W{1}. Therefore, a parameter
dependence of the form $\alpha^{(\Delta+1)\tw}$, as given
by~\cite{DaskalakisP06,GottlobGS05}, is more likely to be close to the correct
answer.

Motivated by this, we then set out to identify the \emph{correct} parameter
dependence as precisely as possible. On the algorithmic side, the
$\alpha^{(\Delta+1)\tw}$ dependence of the algorithms
of~\cite{DaskalakisP06,GottlobGS05} is natural, as it keeps the information
necessary to ensure that when a vertex is forgotten the corresponding player is
locally stable (does not wish to deviate). Despite this, we show that the
natural approach is in fact \emph{not} optimal and its parameter dependence can
be improved by a constant factor in the exponent. Our algorithmic results are as
follows:

\begin{theorem}\label{thm:upper}
  The existence of a PNE in a graphical game $\mathcal{G}$, where players have at
  most $\alpha$ strategies and utilities depending on at most $\Delta$ other
  players, can be decided in time:
  \begin{enumerate}
    \item $\alpha^{\twcoe\tw}\cdot|\mathcal{G}|^{O(1)}$ where $\tw$ is the width
      of a given tree decomposition;
    \item $\alpha^{\pwcoe{}\pw}\cdot|\mathcal{G}|^{O(1)}$ where $\pw$ is the
      width of a given path decomposition;
    \item $\alpha^{\ctw}\cdot|\mathcal{G}|^{O(1)}$ where $\ctw$ is the width of
      a given cutwidth ordering.
  \end{enumerate}
\end{theorem}
The main takeaway is that we decrease the exponent of the parameter dependence
by a factor of roughly $2$ and $\frac{3}{2}$ respectively for pathwidth and
treewidth, which is done by tightening the constants in a well-known
relationship between the widths of a graph and its square.%

Given that our algorithmic improvement is based on an immediate application of a
combinatorial relation, it is then important that we can establish our results
to be (likely) essentially optimal for pathwidth and cutwidth. In particular, we
show the following:

\begin{theorem}\label{thm:lower}
  The following statements are equivalent:
  \begin{enumerate}
    \item The \pwseth{} is false.\label{thm:lower:item-1}
    \item There exists an odd integer $\Delta\ge 3$, an integer $\alpha\ge 2$
      and $\varepsilon>0$ such that the existence of a PNE in a graphical game
      $\mathcal{G}$, where each player has at most $\alpha$ strategies and her
      utility depends on at most $\Delta$ other players can be decided in time
      $(\alpha-\varepsilon)^{\pwcoe{}\pw}\cdot|\mathcal{G}|^{O(1)}$ where $\pw$
      is the width of a given path decomposition.\label{thm:lower:item-2}
    \item There exists an integer $\alpha\ge 2$ and $\varepsilon>0$ such that
      the same problem can be decided in time
      $(\alpha-\varepsilon)^{\ctw}\cdot|\mathcal{G}|^{O(1)}$ where $\ctw$ is
      the width of a given cutwidth ordering.\label{thm:lower:item-3}
  \end{enumerate}
\end{theorem}
The \pwseth{} states that the standard dynamic programming algorithm for
\textsc{SAT} parameterized by pathwidth, which has parameter dependence
$2^{\pw}$, cannot be improved to $(2-\varepsilon)^{\pw}$. Observe that
Theorem~\ref{thm:lower} not only establishes that improving the algorithms of
Theorem~\ref{thm:upper} for pathwidth and cutwidth is impossible under the
\pwseth{}, but in fact shows that such an improvement is \emph{equivalent} to
falsifying the \pwseth{}. The lower bound for pathwidth of course also
automatically applies to treewidth, which leaves it as an intriguing open
question whether we can obtain a better lower bound for treewidth or a better
algorithm.

\subparagraph*{Techniques} With respect to techniques, the main contribution of
this paper is a tightening of a well-known combinatorial relation between the
width of a graph $G$, its maximum degree, and the width of $G^2$. Before we
describe this, let us explain the context through which this question enters the
picture.  We are given as input a \emph{digraph} $G$ with $n$ vertices,
representing the players. A natural way to decide if a PNE exists is to
construct a CSP instance with $n$ variables of domain $\alpha$ (one variable for
each player) and add for each player $u$ a constraint that includes $u$ and all
her out-neighbors and is satisfied only in joint strategies where $u$ does not
wish to deviate.  The primal graph of this instance can be obtained from $G$ by
turning the out-neighborhood of every vertex into an undirected clique; we
denote this graph as $G^{\gsym}$ and call it the \gname{} of $G$. Observe that
if $G$ is undirected (or bi-directed), $G^{\gsym} = G^2$, where vertices at
distance at most $2$ in $G$ are adjacent.

The natural next step is to solve the CSP instance. To this end, if we bound the
treewidth of $G^{\gsym}$, we can use a standard algorithm with dependence
$\alpha^{\tw({G^{\gsym}})}$. The treewidth of $G^{\gsym}$ (and the treewidth of
$G^2$) can easily be upper-bounded by $(\Delta+1)(\tw(G)+1)$ by taking a tree
decomposition of $G$ and adding to each bag that contains a vertex $u$ all
(out-)neighbors of $u$. This is {well-known}, explicitly stated in~\cite[Lemma
4.5]{DaskalakisP06},~\cite[Theorem 16]{GurskiW25},~\cite[Theorem
12]{katsikarelis_improved_2023},~\cite[Lemma
1]{hanaka_computing_nodate},~\cite[Lemma
6.0.2]{bonomo-braberman_new_2022},~\cite[Lemma
4.1]{krumkeModelsApproximationAlgorithms2001}, and used implicitly throughout
the literature, e.g.,~\cite{agrawal_computing_2023}. Furthermore, this upper
bound seems in a sense tight, as $\tw(K_{1,n}^2)=\tw(K_{n+1})=n\cdot
\tw(K_{1,n})$.

Our main algorithmic contribution is to show that the above standard bound is
\emph{not} tight, and can in fact be improved by a constant factor. To the best
of our knowledge, this was not known before, which we find quite surprising
since the combinatorial relation between $\tw(G),\Delta(G)$, and $\tw(G^2)$ is
well-known and has often been used algorithmically as mentioned. More precisely
we prove the following:

\begin{theorem}\label{thm:combi}
  For any digraph $G$ of maximum out-degree $\Delta$ we have:
  \begin{enumerate}
    \item $\tw(G^{\gsym})\le
      \twcoe{}\tw(G)+2\lfloor\frac{2\Delta}{3}\rfloor${;}
    \item $\pw(G^{\gsym})\le
      \pwcoe{}\pw(G)+2\lfloor\frac{\Delta}{2}\rfloor${;}
    \item $\pw(G^{\gsym})\le \ctw(G)+\Delta${.}
  \end{enumerate} %
  Furthermore, the above bounds are constructible.
\end{theorem}
Even though in this paper we focus on digraphs, we observe that we obtain a
corollary for undirected graphs which is likely to be of independent interest.
This follows immediately from Theorem~\ref{thm:combi} by considering an
undirected graph as a bi-directed digraph.

\begin{corollary}\label{cor:combi-undir}
  For any undirected graph $G$ of maximum degree $\Delta$ we have:
  \begin{enumerate}
    \item $\tw(G^{2})\le \twcoe{} \tw(G)+2\lfloor\frac{2\Delta}{3}\rfloor${;}
    \item $\pw(G^{2})\le \pwcoe{} \pw(G)+2\lfloor\frac{\Delta}{2}\rfloor${.}
  \end{enumerate}
\end{corollary}

Interestingly, we are able to establish for both pathwidth and treewidth that
the results of Corollary~\ref{cor:combi-undir} (and therefore also of
Theorem~\ref{thm:combi}) are essentially tight by giving an infinite family of
examples. Therefore, it is not possible to squeeze any further improvement by
bounding the width of $G^{\gsym}$ (or $G^2$) and our algorithms are optimal
\emph{under this approach}.

To establish the optimality of our algorithms \textit{in the general sense}, we
give reductions from CSPs of appropriate alphabets for the pathwidth and
cutwidth cases, showing that improving our algorithms would falsify the
\pwseth{}. Notably, one advantage of our algorithmic approach of reducing to a
CSP is that this also automatically implies a reduction in the other direction,
showing that improving our algorithms is \emph{equivalent} to falsifying the
\pwseth{}.

Finally, let us mention that we focus on digraphs, because this allows a more
fine-grained exploration of the problem's complexity, but in general we refer to
undirected graph width parameters. These correspond to the width of the
underlying graph, that is, the graph obtained by ignoring directions. Focusing
on digraphs makes our algorithmic results stronger, as the previous works giving
the less efficient $\alpha^{(\Delta+1)\tw}$ parameter
dependence~\cite{DaskalakisP06,GottlobGS05} were focusing on the special case
when the input is undirected (bi-directed). Conversely, it is worth noting that
our reduction showing that the problem is \W{1}-hard parameterized by treewidth
produces bi-directed graphs, which is exactly the case purportedly solved
by~\cite{thomas_pure_2015} in \FPT{} time.

\subparagraph*{Other related work} The problem of computing PNEs is a
well-studied topic and computing equilibria is generally hard, even in the
parameterized setting and especially if one is also looking to optimize some
objective
~\cite{BiloM21,ConitzerS08,ElkindGG07,GrecoS09,HanakaL22,Lampis21,Peters16a,SchoenebeckV12}.
Mixed equilibria, which are guaranteed to exist, are also \PPAD-hard to compute
even for graphical games of pathwidth $2$~\cite{ElkindGG06}. In this paper we
focus on undirected structural parameters, even though the input is a digraph,
as directed analogues of treewidth are known to be unhelpful for this problem
~\cite{JiangS10}.

The \pwseth{} was recently put forward by Lampis~\cite{lampis_primal_2025} as a
more solid alternative to the SETH and shown to be equivalent to numerous tight
lower bounds for problems parameterized by pathwidth. It is implied by weaker
versions of the SETH and the Set Cover Conjecture~\cite{Lampis26a}, and has
recently been used as a starting point for tight lower
bounds~\cite{EsmerM25,abs-2506-01645,HartmannM25,LampisV25,schepper2025faster}.
We also remark that the vast majority of problems have the same complexity
parameterized by pathwidth and
treewidth~\cite{BelmonteKLMO22,CyganKN18,Lampis23}, making the gap we have
between $\frac 12$ and $\frac 23$ more intriguing.

\subparagraph*{Paper organization} After some preliminaries in
Section~\ref{sec:prelim}, we explain in Section~\ref{sec:thomas_wrong} why the
results of~\cite{thomas_pure_2015} are flawed. Then, in
Section~\ref{sec:square_theorems} we present our combinatorial bounds, their
algorithmic applications, as well as examples showing the tightness. In
Section~\ref{sec:seth_lower_bounds} we present our lower bounds based on the
\pwseth{} and close with some discussion in Section~\ref{sec:conclusion}.

\section{Preliminaries}\label{sec:prelim}

A graphical game is a tuple $\mathcal{G}=(G,\{M_v\}_{v\in V(G)})$. $G$ is the
digraph with each vertex corresponding to one player and each arc indicating a
utility dependency. Let $N^{+}(v)$ be the set of out-neighbors of $v$,
$N^{+}[v]=N^{+}(v)\cup\{v\}$ be the closed out-neighborhood of $v$. Then the
utility of player $v$ is determined by the strategies of $N^{+}[v]$, and is
given by the payoff matrix $M_v$. Let $\st(v)$ be the set of available
strategies for player $v$, then $M_v$ is a matrix mapping $\st(v)\times
\prod_{u\in N^{+}(v)}\st(u)$ to $\mathbb{R}$. Given a joint strategy $\mathbf s$
of $V(G)$, for each player $v$, we say $v$ is \textit{locally stable} if
changing her strategy will not increase the utility, i.e.,
\begin{equation}
  \forall x\in \st(v), \quad M_v(\mathbf s|_{v},\mathbf s|_{N^{+}(v)})\ge M_v(x,\mathbf s|_{N^{+}(v)}),\label{eq:vertex-stable}
\end{equation}
where $\mathbf s|_v, \mathbf s|_U$ denote the (joint) strategy of a player $v$,
a set of players $U$ in $\mathbf s$ respectively. Our problem is to decide
whether a given graphical game admits a pure Nash equilibrium (PNE), that is, a
joint strategy of all players in which everyone is locally stable. For
simplicity, we say a graphical game is an $\alpha$-strategy game if every player
has at most $\alpha$ strategies. Moreover, any structural properties (e.g.,
width, maximum degree) of a game implicitly refer to its underlying game graph.
Now we give the remaining two key concepts in our work.
\begin{definition}[\gname]
  Given a digraph $G$, we define the \gname{} of $G$, denoted by $G^{\gsym}$, as
  the undirected graph with the same vertex set as $G$ and an edge between two
  vertices $u$ and $v$ if and only if there is a vertex $w$ such that $\{u,v\}
  \subseteq N^{+}[w]$.
\end{definition}
Note that \gname{} is a generalization of the square graph to digraphs,
and the moral graph~\cite{verma_deciding_1993} to cyclic digraphs.

\begin{definition}[Constraint satisfaction problem]
  A constraint satisfaction problem is defined as a triple $\langle X, \Sigma, C
  \rangle$, where $X$ is a finite set of variables, $\Sigma$ is a finite
  alphabet, and $C$ is a finite set of constraints of the form $\langle t_j, R_j
  \rangle$, where $t_j$ is a tuple of distinct variables from $X$ and $R_j
  \subseteq \Sigma^{|t_j|}$. The problem asks whether there exists an assignment
  $\sigma: X \to \Sigma$ such that for every constraint $\langle t_j, R_j
  \rangle \in C$ with $t_j = (x_1, \dots, x_{|t_j|})$ we have $(\sigma(x_1),
  \dots, \sigma(x_{|t_j|})) \in R_j$.
\end{definition}

Finally, recall that when we talk about the width of a digraph, we are referring
to the width of the underlying undirected graph. This graph is not necessarily
simple, as if both arcs $(u,v)$ and $(v,u)$ exist, they form parallel edges in
the underlying undirected graph. Multiple edges do not affect the pathwidth and
treewidth, but they do affect the cutwidth as it counts the number of edges
instead of vertices.

\section{W[1]-hardness parameterized by vertex cover}\label{sec:thomas_wrong}

In this section, we explain why the algorithmic result of
~\cite{thomas_pure_2015} is flawed. They claim:
\begin{claim}[{\cite[Theorem 2]{thomas_pure_2015}}]\label{clm:tw-alg}
  The existence of a PNE in an \textit{undirected} $\alpha$-strategy graphical
  game $\mathcal{G}$ with treewidth $\tw$ can be decided in time
  $\alpha^{\tw}|\mathcal{G}|^{O(1)}$.
\end{claim}%
In contrast, we prove the problem is \W{1}-hard parameterized by vertex cover,
which is an upper bound of treewidth, and therefore the \W{1}-hardness also
applies to treewidth.
\begin{theorem}\label{thm:plus-dependency-impossible}
  Determining the existence of a PNE in an undirected 2-strategy
  graphical game is \W{1}-hard parameterized by vertex
  cover.
\end{theorem}
\begin{proof}
  We reduce from Multicolored Clique. Given a $k$-partite graph $G$ with vertex
  partition $V_1\cup V_2\cup \dots \cup V_k$ and each $V_i$ is of size $n$, we
  build an undirected graphical game $(G',\{M_v\}_{v\in V(G')})$ where $G'$ has
  a vertex cover of size at most $3k^2$ such that there exists a PNE in the game if and
  only if there exists a clique of size $k$ in $G$.

  For each part $V_i$, we create an independent set of size $\ell=\lceil\log_2
  n\rceil$, $v_{i,1},v_{i,2},\dots,v_{i,\ell}$. The joint strategy of these
  $\ell$ players encodes a choice of vertex in $V_i$. Then for every pair of
  $V_{i_1},V_{i_2} (i_1<i_2)$, we create three players
  $u_{i_1i_2,1},u_{i_1i_2,2},u_{i_1i_2,3}$, and add edges from $u_{i_1i_2,1}$ to
  $v_{i,j}$ for all $i\in \{i_1,i_2\}, j\in[\ell]$, and set up
  $u_{i_1i_2,2},u_{i_1i_2,3}$ as a \textit{conditional} matching pennies gadget
  forcing $u_{i_1i_2,1}$ to play $1$\footnote{%
    Matching pennies is a 2-player game with no PNE. Here if $u_{i_1i_2,1}$ does
    not play $1$, $u_{i_1i_2,2},u_{i_1i_2,3}$ will simulate a matching pennies
    game, denying any PNE, therefore forcing $u_{i_1i_2,1}$ to play $1$.%
  }. %
  $M_{u_{i_1i_2,1}}$ is defined such that it plays $1$ if and only if the joint
  strategy of $v_{i_1,1},\dots,v_{i_1,\ell},v_{i_2,1},\dots,v_{i_2,\ell}$
  encodes a pair of vertices that are adjacent in $G$, and $M_{v_{i,j}}$ is a
  constant matrix.
  Note that each player in this graph interacts with at most
  $\max\{2\ell+2,k-1\}$ other players, so the construction of the payoff
  matrices can be done in time $2^kn^{O(1)}$, resulting in an \FPT{}
  reduction. Finally, it is easy to see $\{u_{i_1i_2,h}\}_{1\le i_1<i_2\le
  k,h\in [3]}$ is a vertex cover of size at most $3k^2$. So if there is an \FPT{}
  algorithm for PNE
  parameterized by vertex cover, then we can decide the existence of a clique
  of size $k$ in \FPT{} time.
\end{proof}
Besides \cref{clm:tw-alg}, we observe that~\cite[Theorem 1]{thomas_pure_2015}
contains a second flawed algorithmic claim. In particular, they claim to give a
polynomial-time algorithm when the maximum number of strategies $\alpha=2$
(even for games of unbounded degree and treewidth).  However, this is not
possible (unless \P=\NP) due to {\cite[Theorem 6.2]{SchoenebeckV12}}.
\begin{theorem}[{\cite[Theorem 6.2]{SchoenebeckV12}}]%
  Determining the existence of a PNE in an undirected graphical game is
  \NP-complete, even if it is $2$-strategy with maximum degree $3$.
\end{theorem}
Besides the hardness result, we give some further notes in
Appendix~\ref{app:notes-on-thomas}.

\begin{toappendix}

  \appendixsection{Further Notes on~\cite{thomas_pure_2015}}\label{app:notes-on-thomas}

  To complete the discussion of Section~\ref{sec:thomas_wrong} let us also
  roughly explain where the error lies in the analysis of their treewidth
  algorithm. They claim that the main idea of their dynamic program is to record
  for each vertex in the bag not only its strategy in the current partial
  solution, but also its optimal response to the strategies of its forgotten
  neighbors. Intuitively, this seems unlikely to work as in order to calculate
  the best response, one would also need to take into account the strategies of
  neighbors that have not yet been introduced. Furthermore, in their analysis
  they claim that when a vertex is forgotten we update the payoff matrices of
  its neighbors to only keep information consistent with the strategy of the
  forgotten vertex. This would, however, need to keep a different version of the
  payoff matrix for each combination of strategies of the forgotten neighbors of
  a vertex, which is not taken into account in the running time analysis.

  To be more concrete, consider the following game graph and an associated tree
  decomposition:
  \begin{center}
    \begin{tikzpicture}
      \begin{scope}[rotate=45]
        \node[fill vertex, label=below:1] (1) at (0,0) {};
        \node[fill vertex, label=left:2] (2) at (0,1) {};
        \node[fill vertex, label=right:3] (3) at (1,0) {};
        \node[fill vertex, label=above:4] (4) at (1,1) {};
        \draw[edge] (1) -- (2) -- (4) -- (3) -- (1);
        \draw[edge] (2) -- (3);
      \end{scope}

      \begin{scope}[xshift=4cm,yshift=-1.3cm,bag/.style={draw=mDark, ellipse, minimum width=1cm, minimum height=0.3cm, font=\scriptsize},yscale=1]
        \node[bag] (2) at (0,1) {$\{1,2,3\}$};
        \node[bag] (3) at (0,2) {$\{2,3\}$};
        \node[bag] (4) at (0,3) {$\{2,3,4\}$};
        \draw (2) -- (3) -- (4);
      \end{scope}
    \end{tikzpicture}
  \end{center}
  Each player has two strategies. The payoff matrices of the players are defined
  so that $2$ wants to play the same strategy as $3$ if $1$ and $4$ play the
  same strategy and wants to play the opposite strategy of $3$ otherwise, while
  $3$ wants to play the same strategy as $2$ if $1$ and $4$ play different
  strategies and wants to play the opposite strategy of $2$ otherwise. $1$ and
  $4$ have constant payoff matrices. Then when $1$ is forgotten, no constraint
  is added to the choice of $2$ and $3$. Therefore, when $4$ is introduced,
  $2,3$ and $4$ have no constraint on their strategies, and the algorithm would
  conclude that there is a PNE, which is apparently not the case.

\end{toappendix}

\section{On the width of (generalized) graph squares}

In this section, we prove our upper bounds on the different width parameters of
$G^{\gsym}$, yielding three algorithms for computing PNEs parameterized by $\pw,\tw,\ctw$ respectively. After that, we give examples
showing our bounds are essentially tight for $\pw$
and $\tw$.

\subsection{Upper bounds}\label{sec:square_theorems}

We present three combinatorial upper bounds on the width of $G^{\gsym}$ as a
function of the maximum out-degree and the width of $G$. Let us first give some
intuition starting from pathwidth. Recall that it is easy to establish that
$\pw(G^{\gsym})\le (\Delta+1)(\pw(G)+1)$ simply by adding all the out-neighbors
of any vertex to each bag that contains it. Our approach starts with an attempt
of not adding all the out-neighbors. For each vertex $u$, and for each bag $X$
that contains $u$, either the minority (at most $\lfloor
\frac{\Delta}{2}\rfloor$) of its out-neighbors have been introduced in bags
before $X$; or the minority of its out-neighbors are introduced in later bags.
In either case, we add to $X$ whichever subset of out-neighbors forms the
minority. This gives us a decomposition of width at most
$(1+\lfloor\frac{\Delta}{2}\rfloor)(\pw(G)+1)$, with a significant problem that
no bag contains all the closed out-neighborhood of $u$, which induce a clique in
$G^{\gsym}$. To fix this, we observe that $u$ transitions from having the
minority to the majority of its neighbors already introduced in \emph{a pair of
consecutive bags}. We can therefore insert a bag between the two, adding all the
out-neighbors of $u$ and paying a second $\lfloor \frac{\Delta}{2}\rfloor$
additive term.

The trick above also generalizes to treewidth, but the perfect balance of the
out-neighborhood between the ``left'' and ``right'' parts is no longer possible.
In contrast, the best to achieve with this trick is to ``save'' one third of the
out-neighborhood of each vertex. Intuitively, this is inevitable because there
may be a join bag (a bag of degree $3$) in which each vertex has \textit{one
third of its out-neighbors} in each of the three directions in the
decomposition.

Finally, the bound for cutwidth is more straightforward %
where the path decomposition is constructed by adding to each bag the heads of
arcs crossing a cut in a cutwidth ordering.

\begin{theorem}\label{thm:pw-sqr-ub}
  For a digraph $G$ with maximum out-degree $\Delta$ we have
  \begin{equation}
    \pw(G^{\gsym})\le \pwcoe{}\pw(G)+2\lfloor\textstyle\frac{\Delta}{2}\rfloor.\label{eq:pw-sqr-ub}
  \end{equation}
  Furthermore, given a path decomposition of $G$ of width $\pw$, a path
  decomposition of $G^{\gsym}$ of width at most
  $\pwcoe{}\pw+2\lfloor\textstyle\frac{\Delta}{2}\rfloor$ can be found in
  polynomial time.
\end{theorem}

\begin{proof}
  Assume we are given a nice path decomposition $\mathcal{P}=(X_1, X_2, \ldots,
  X_r)$ of $G$ of width $\pw$. We will construct a path decomposition
  $\mathcal{P}^{\gsym}$ for $G^{\gsym}$ with the claimed width. Let $\lambda=\lceil{\frac {\Delta}2}\rceil$ be a threshold
  parameter. We scan the bags
  from left to right, then in each bag $X_i$ and for a vertex $v\in X_i$,
  $N^{+}(v)$ can be bipartitioned as follows:
  \begin{enumerate}
    \item $\pre_i(v)$: the \textit{present neighbors} of $v$ that are either in
      the current bag or have been forgotten; %
    \item $\fut_i(v)$: the \textit{future neighbors} of $v$ that are still to be
      introduced in the future.
  \end{enumerate}
  For a vertex $v$, suppose it is introduced and forgotten in bags $i_v$ and
  $f_v+1$, respectively. Then,
  \[%
    \pre_{i_v}(v)  =X_{i_v}\cap N^{+}(v),  \qquad \pre_{f_v}(v)  =N^{+}(v),
    \qquad \pre_{i}(v)    \subseteq \pre_{i+1}(v), \forall i_v\le i < f_v.
  \]%
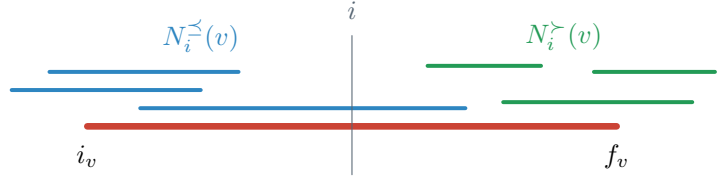
  \begin{center}
    \begin{tikzpicture}[line cap=round, yscale=0.8]
      \draw[mRose, line width=2.5pt] (-3.5,0) -- (3.5,0);
      \node at (-3.5,-0.5) {$i_v$};
      \node at (3.5,-0.5) {$f_v$};

      \draw[mBlue, line width=1.5pt] (-4, 0.9) -- (-1.5, 0.9);
      \draw[mBlue, line width=1.5pt] (-4.5, 0.6) -- (-2, 0.6);
      \draw[mBlue, line width=1.5pt] (-2.8, 0.3) -- (1.5, 0.3);

      \node[mBlue] at (-2,1.5) {$\pre_i(v)$};
      \node[mGreen] at (2.8,1.5) {$\fut_i(v)$};

      \draw[mGreen, line width=1.5pt] (1, 1) -- (2.5, 1);
      \draw[mGreen, line width=1.5pt] (3.2, 0.9) -- (4.8, 0.9);
      \draw[mGreen, line width=1.5pt] (2, 0.4) -- (4.5, 0.4);
      \draw[mDark] (0,-0.8) -- (0,1.5) node[above=0.1cm] {$i$};
    \end{tikzpicture}
  \end{center}
  The statements above follow from the fact that every neighbor of $v$ must appear together with $v$ in a bag, and from the connectivity constraint on bags that contain each specific vertex in a path decomposition.
  We say a vertex $v$ is \textit{heavy} in bag $X_i$ if $v\in X_i$ and
  $|\pre_i(v)|\ge \lambda$. Let $L_i\subseteq X_i$ be the set of vertices
  that are \textit{light} in $X_i$. Note that the transition from light to heavy
  is a one-way process because $|\pre_i(v)|$ is non-decreasing as $i$ increases.

  Now we scan the bags of $\mathcal{P}$ from left to right, and {iteratively}
  append new bags to $\mathcal{P}^{\gsym}$. For each bag $X_i$, we let $B_{i,1},
  B_{i,2}, \dots, B_{i,m_i}$ be the sequence of bags we append when processing
  $X_i$. For $X_i$ introducing a vertex $v$, let $t_1,t_2,\dots,t_{\ell}$ be the
  list of vertices that become heavy in $X_i$. It is worth noting that $v$ is
  also possibly heavy, if sufficiently many neighbors of $v$ are already in
  $X_{i-1}$. We now create a sequence of $m_i=2\ell+1$ bags. For the $2k$-th
  bag $(k\in [\ell])$, we introduce all vertices in $N^{+}[t_k]$, and in the
  subsequent bag, we forget every present neighbor of $t_k$ that is already
  forgotten and does not belong to the present neighborhood of some other
  (temporarily) \textit{light} vertex in
  $L_i\cup\{t_{k+1},\dots,t_{\ell}\}$\footnote{We abuse the notation and also
  call vertices in $t$ that have not yet been processed as {temporarily light}.}.
  Formally (let $B_{i,1}=B_{i-1,m_{i-1}}\cup \{v\}$),%
  \begin{align}
    \forall 1\le k\le \ell: \quad B_{i,2k}  = & B_{i,2k-1}\cup \fut_i(t_k)                                                                   \label{eq:balancing-bags}            \\
    \quad B_{i,2k+1}    =                         & B_{i,2k}\setminus \left(\pre_i(t_k)\setminus \pre_i(L_i\cup \{t_{k+1},\dots,t_{\ell}\})\right),\label{eq:balancing-forget-bags}
  \end{align}
  where $\pre_i(S)$ is defined as $\bigcup_{u\in S}\pre_i(u)$. For $X_i$
  forgetting $v$, if $v\in \pre_i(L_i)$, we simply neglect this bag, i.e.,
  $m_i=0$ and let $B_{i,0}=B_{i-1,m_{i-1}}$; otherwise, we append a new bag that
  forgets $v$, i.e., $m_i=1$ and $B_{i,1}=B_{i-1,m_{i-1}}\setminus \{v\}$. We
  claim for each $i\in [|\mathcal{P}|]$,
  \(
    B_{i,m_i}=X_i\cup \bigcup_{v\in L_i} \pre_i(v) \cup \bigcup_{v\in X_i\setminus L_i} \fut_i(v),
  \)
  which shows such bags have size at most $\pwcoe{}(\pw+1)$, and the same bound
  can also be shown for $B_{i,2k+1}$, and thus the size of $B_{i,2k}$ can have
  an additional $O(\Delta)$ term, giving the desired bound. Full proof of
  validity and width analysis can be found in Appendix~\ref{app:pw-sqr-ub}.
\end{proof}

\begin{toappendix}

  \appendixsection{The validity and width of the construction in
  \cref{thm:pw-sqr-ub}}\label{app:pw-sqr-ub}

  We first show the $\mathcal{P}^{\gsym}$ constructed by the procedure in the
  proof of \cref{thm:pw-sqr-ub} is a valid path decomposition for $G^{\gsym}$.
  It is easy to see every vertex is covered. For an edge $(u,v)\in E(G^{\gsym})$,
  either $(u,v)\in E(G)$, or there exists $w$ such that $u,v\in N^{+}(w)$. In
  the first case, this edge is covered by \cref{eq:balancing-bags} when $t_k=u$.
  In the second case, this edge is covered by \cref{eq:balancing-bags} when
  $t_k=w$. Finally, to show every vertex appears in a contiguous segment, it
  suffices to prove that each vertex is only forgotten once. If $v$ is forgotten
  in $B_{i,*}$ for some $i$ in $\mathcal{P}^{\gsym}$, then either $X_i$ forgets
  $v$ while $v\not\in \pre_i(L_i)$; or $X_i$ is an introduce bag where $v$ is no
  longer in $\pre_i(L_i)$ with $v$ already forgotten. To conclude, $v$ is
  forgotten in $B_{i,*}$ if and only if $X_i$ is a bag that
  \[
    v\not\in X_i \wedge v\not\in \pre_i(L_i) \wedge (v\in X_{i-1} \vee v\in \pre_{i-1}(L_{i-1})).
  \]
  Suppose there exists another bag $X_{i'}$ satisfying the constraint above and
  without loss of generality we assume $i<i'$. $v$ is already forgotten in
  $X_i$, %
  so $v\not\in X_{i'-1}$, therefore $v\in \pre_{i'-1}(L_{i'-1})$ must hold. Let
  $u\in X_{i'-1}$ be a vertex such that $v\in \pre_{i'-1}(u)$. Suppose $u$ is
  introduced before $X_i$, then since $v\in \pre_{i'-1}(u)$, we also have $v\in
  \pre_i(u)$, and thus $v\in \pre_i(L_i)$, which contradicts the assumption.
  Otherwise, suppose $u$ is introduced at or after $X_{i}$, then since $v$ has
  been forgotten, the edge $(u,v)$ is not covered by any bag. We thus proved such
  $X_i$ is unique, and each vertex is only forgotten once.

  Now it remains to analyze the width of $\mathcal{P}^{\gsym}$. It is easy to
  see
  \begin{claim*}
    For each $i\in [|\mathcal{P}|]$, we have
    \[
      B_{i,m_i}=X_i\cup \bigcup_{v\in L_i} \pre_i(v) \cup \bigcup_{v\in X_i\setminus L_i} \fut_i(v).
    \]
  \end{claim*}
  Then it follows that
  \[
    \begin{aligned}
      |B_{i,m_i}|\le & (\pw+1)+(\lambda-1)|L_i|+(\Delta-\lambda)(\pw+1-|L_i|) \\
      =              & (\Delta-\lambda+1)(\pw+1)+(2\lambda-\Delta-1)|L_i|     \\
      \le            & \pwcoe{} (\pw+1),
    \end{aligned}
  \]
  For intermediate bags, we have
  \begin{claim*}
    For each $i\in [|\mathcal{P}|]$ s.t. $m_i\ge 3$, let $\ell=\frac{m_i-1}2$.
    Then for any $k(0\le k \le \ell)$, %
    \[
      B_{i,2k+1}=X_i\cup \bigcup_{v\in L_i\cup \{t_{k+1},\dots,t_{\ell}\}} \pre_i(v) \cup \bigcup_{v\in X_i\setminus L_i\setminus \{t_{k+1},\dots,t_{\ell}\}} \fut_i(v).
    \]
  \end{claim*}
  Since $\pre_i(t_{*})$ intersects $X_i$ in at least one vertex ($v$), we have $|\pre_i(t_{*})\setminus X_i|\le
  \lambda-1$, giving
  \[
    \begin{aligned}
      |B_{i,2k+1}|\le & (\pw+1)+(\lambda-1)(|L_i|+\ell-k)+(\pw+1-(|L_i|+\ell-k))(\Delta-\lambda) \\
      \le           & \lfloor\textstyle\frac {\Delta}{2}+1\rfloor (\pw+1).
    \end{aligned}
  \]
  It then follows that
  \[
    \begin{aligned}
      |B_{i,2k}|\le & |B_{i,2k-1}|+\Delta-\lambda                                                                   \\
      \le             & \lfloor\textstyle\frac {\Delta}{2}+1\rfloor (\pw+1)+\lfloor\textstyle\frac{\Delta}{2}\rfloor.
    \end{aligned}
  \]

\end{toappendix}

\begin{theorem}\label{thm:tw-sqr-ub}
  For a digraph $G$ with maximum out-degree $\Delta$, we have
  \begin{equation}
    \tw(G^{\gsym})\le \twcoe{}\tw(G)+2\lfloor \textstyle \frac{2\Delta}{3}\rfloor.
  \end{equation}
  Furthermore, given a tree decomposition of $G$ of width $\tw$, a tree
  decomposition of $G^{\gsym}$ of width at most $\twcoe{}\tw+2\lfloor \textstyle
  \frac{2\Delta}{3}\rfloor$ can be found in polynomial time.
\end{theorem}

\begin{proof}
  Let $\lambda=\lceil\frac{\Delta}{3}\rceil$. Given a nice tree decomposition
  $\mathcal{T}=(T,\{X_i\}_{i\in V(T)})$, we traverse the bags in a bottom-up
  manner. For each bag $X_i$ and for each vertex $v$ in this bag, we define
  $\pre_i(v)$, $\fut_i(v)$ and the notion of \textit{heavy} as in
  \cref{thm:pw-sqr-ub}. For each vertex $v$, if it is introduced in a bag
  $X_{i_v}$ and forgotten in the parent of bag $X_{f_v}$, we have
  \[
    \pre_{i_v}(v)=  X_{i_v}\cap N^{+}(v), \qquad \pre_{f_v}(v)=  N^{+}(v),  \qquad \pre_{i}(v)\subseteq  \pre_{p(i)}(v), \forall i_v\preceq i \prec f_v,
  \]
  where $p(i)$ is the parent of $i$, and $i\preceq j$ ($i\prec j$) denotes that
  $j$ is an ancestor (a proper ancestor) of $i$. The bags where $v$ is heavy
  form a subtree rooted at $f_v$ (denoted by $T[v|\mathsf{heavy}]$), and is a
  part of the subtree of all bags containing $v$ (denoted by $T[v]$). We
  arbitrarily fix a unique leaf in $T[v|\mathsf{heavy}]$, and call it the
  \textit{switching leaf} of $v$. Moreover, we call the path from the switching
  leaf to $f_v$ the \textit{switching path} of $v$. Then for each vertex $v$ in
  a bag $X_i$, we introduce another notion of \textit{switched}: $v$ is switched
  in $X_i$ when $X_i$ belongs to the switching path of $v$. And we let
  $L_i\subseteq X_i$ to denote the set of vertices that are \textit{unswitched}
  in $X_i$.
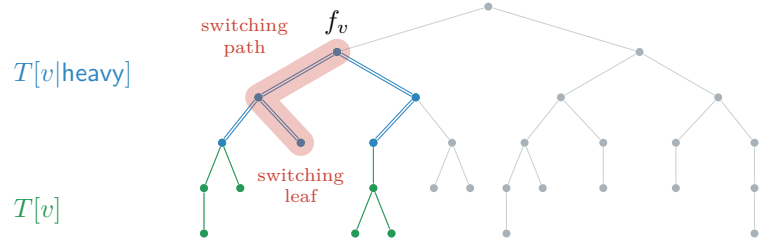
  \begin{center}
    \begin{tikzpicture}[xscale=8,yscale=3]
      \foreach \i/\x/\y in { 0/0.50/0.00, 1/0.25/-0.20, 5/0.12/-0.40,
        7/0.06/-0.60, 19/0.03/-0.80, 25/0.03/-1.00, 29/0.09/-0.80,
        27/0.19/-0.60, 6/0.38/-0.40, 8/0.31/-0.60, 10/0.31/-0.80, 16/0.28/-1.00,
        17/0.34/-1.00, 13/0.44/-0.60, 21/0.41/-0.80, 22/0.47/-0.80,
        2/0.75/-0.20, 3/0.62/-0.40, 4/0.56/-0.60, 14/0.53/-0.80, 24/0.53/-1.00,
        18/0.59/-0.80, 15/0.69/-0.60, 20/0.69/-0.80, 9/0.88/-0.40,
        11/0.81/-0.60, 12/0.81/-0.80, 23/0.94/-0.60, 26/0.94/-0.80,
      28/0.94/-1.00}{%
        \node[fill vertex, opacity=0.5] (v\i) at (\x, \y) {};
      }

      \foreach \u/\v in { 0/1, 0/2, 1/5, 1/6, 2/3, 2/9, 3/4, 3/15, 4/14, 4/18,
        5/7, 5/27, 6/8, 6/13, 7/19, 7/29, 8/10, 9/11, 9/23, 10/16, 10/17, 11/12,
      13/21, 13/22, 14/24, 15/20, 19/25, 23/26, 26/28}{%
        \draw[edge, opacity=0.5] (v\u) -- (v\v);%
      }

      \foreach \i in {5,7,19,25,29,27,1,6,8,10,16,17}{
        \node[fill vertex, mGreen] at (v\i) {};
      }
      \foreach \u/\v in { 1/5, 1/6, 5/7, 5/27, 6/8, 7/19, 7/29, 8/10, 10/16, 10/17, 19/25}{%
        \draw[edge, mGreen] (v\u) -- (v\v);%
      }

      \foreach \i in {5,7,27,1,6,8}{
        \node[fill vertex, mBlue] at (v\i) {};
      }
      \foreach \u/\v in { 1/5, 1/6, 5/7, 5/27, 6/8}{%
        \draw[edge, mBlue, double] (v\u) -- (v\v);%
      }

      \draw[mRose, opacity=0.3, line width=10pt, line cap=round, line join=round]
      (v1.center) -- (v5.center) -- (v27.center);
      \node[yshift=-16pt] at (v27){%
        \begin{varwidth}{3cm}
          \color{mRose}\centering \scriptsize switching \\ leaf
        \end{varwidth}%
      };%
      \node[xshift=-20pt,yshift=14pt] at ($(v1)!0.5!(v5)$){%
        \begin{varwidth}{3cm}
          \color{mRose}\centering \scriptsize switching \\ path
        \end{varwidth}%
      };;

      \node[anchor=west] at (-0.3,-0.3){\color{mBlue}$T[v|\mathsf{heavy}]$};
      \node[anchor=west] at (-0.3,-0.9){\color{mGreen}$T[v]$};
      \node[anchor=south, yshift=3pt] at (v1) {$f_v$};
    \end{tikzpicture}
  \end{center}
  Now we are ready to construct a tree decomposition $\mathcal{T}^{\gsym}$ for
  $G^{\gsym}$. For each bag $X_i\in
  \mathcal{T}$, we will construct a sequence of bags that forms a path in
  $\mathcal{T}^{\gsym}$, and denote them by $B_{i,1},B_{i,2},\dots,B_{i,m_i}$ in
  a bottom-up order, and connect $B_{i,m_i}$ to $B_{p(i),1}$ if $i$ is not the
  root.

  For introduce and forget nodes, we use almost identical constructions as
  \cref{thm:pw-sqr-ub}. The only difference here is for introduce nodes we apply
  the construction with \textit{the list of vertices having $X_i$ as their
  switching leaf} instead of the list of vertices that become heavy.%

  For join node $i$ with children $i_1,i_2$, we first create $B_{i,1}$ as
  $B_{i_1,m_{i_1}}\cup B_{i_2,m_{i_2}}$. Note that for a vertex $v$, its future
  neighbors in $X_{i_1}$ is composed of its future neighbors in $X_i$ and
  \textit{present neighbors} in $X_{i_2}$, and the latter part should be removed
  from the current union. $B_{i,2}$ is thus created by iterating through all
  $v\in X_i\setminus L_{i_b} (b\in \{1,2\})$ and forgetting every element in
  $\pre_{i_{3-b}}(v)$ that does not belong to the present neighborhood of some
  other (temporarily) unswitched vertex. Then, for the vertices
  $t_1,\dots,t_{\ell}$ with $X_i$ as the switching leaf, we create bags as
  introduce node case. Namely, we have
  \[
    \begin{aligned}
      B_{i,1}=                                    & B_{i_1,m_{i_1}}\cup B_{i_2,m_{i_2}},  \\                                                                                     %
      B_{i,2}=                                    & B_{i,1}\setminus \bigcup_{b\in \{1,2\},v\in X_i\setminus L_{i_b}} \left(\pre_{i_{3-b}}(v)\setminus \pre_i(L_i\cup \{t_1,\dots,t_{\ell}\})\right), \\%
      \forall 1\le k\le \ell: \quad B_{i,2k+1}  = & B_{i,2k}\cup \fut_i(t_k),   \\                                                                %
      \quad B_{i,2k+2}    =                       & B_{i,2k+1}\setminus \left(\pre_i(t_k)\setminus \pre_i(L_i\cup \{t_{k+1},\dots,t_{\ell}\})\right).%
    \end{aligned}
  \]

  Now we show $\mathcal{T}^{\gsym}$ is valid. The non-trivial part is to show
  the bags including a certain vertex induce a subtree. It suffices to show that
  each vertex is forgotten only once. If $v$ is forgotten in $B_{i,*}$ for some
  $i$, then either $X_i$ is a bag forgetting $v$ while $v\not\in \pre_i(L_i)$;
  or $X_i$ is an introduce or a join bag where $v$ is no longer in $\pre_i(L_i)$
  with $v$ already forgotten. Namely, $v$ is forgotten in $B_{i,*}$ if and only
  if $X_i$ is the bag that $v\not\in X_i\wedge v\not\in \pre_{i}(L_i)$ and
  \[
    \begin{cases}
      v\in \pre_{i_1}(L_{i_1})\vee v\in \pre_{i_2}(L_{i_2}), & \text{if $i$ is a join node with children $i_1$ and $i_2$}, \\                                                             v\in X_{i_1}\vee v\in \pre_{i_1}(L_{i_1}), &\text{if $i$ is an introduce or forget node with child $i_1$.}
    \end{cases}
  \]
  Then we prove by contradiction that such $i$ is unique. We first claim:
  \begin{claim}
    For any bag $X_i$, if $v\in \pre_{i}(L_i)$, then $\exists j\preceq i$ such
    that $v\in X_j$.\label{clm:pre-s-in-bag}
  \end{claim}
  By \cref{clm:pre-s-in-bag}, we can see that there exists $j\preceq i$ such
  that $v\in X_j$, and that $v\not\in X_i$ is equivalent to $f_v\preceq i$ in the condition for vertex $v$ being forgotten in $B_{i,*}$. %
  Suppose there is another introduce or forget node $i'$ with child $i'_1$
  satisfying the same condition, and without loss of generality, we assume
  $f_v\preceq i\preceq i'_1\prec i'$. It immediately follows that $v\not\in
  X_{i'}$, and thus $v\in \pre_{i'_1}(L_{i'_1})$. Let $u\in L_{i'_1}$ be a
  vertex such that $v\in \pre_{i'_1}(u)$, then if $u$ is already introduced in
  $X_{i}$, then $v\in \pre_{i}(u)\subseteq \pre_{i}(L_i)$, contradicting the
  assumption. Otherwise, edge $(u,v)$ is not covered by any bag, resulting in a
  violation. Now we suppose there is another join node $i'$ satisfying the
  condition with children $i'_1,i'_2$, and without loss of generality assume
  $f_v \preceq i\preceq i_1'\prec i'$. Same argument as the previous case rules
  out the possibility of $v\in\pre_{i'_1}(L_{i'_1})$. If
  $v\in\pre_{i'_2}(L_{i'_2})$, then by \cref{clm:pre-s-in-bag} there exists
  $j\preceq i'_2$ such that $v\in X_j$, which contradicts the requirement of
  bags containing $v$ inducing a connected subtree. Therefore, such $i$ is
  unique, and thus $v$ is only forgotten once.

  It remains to analyze the width of $\mathcal{T}^{\gsym}$. We claim for each
  $i\in V(T)$, \(B_{i,m_i}=X_i\cup \bigcup_{v\in L_i}\pre_i(v)\cup \bigcup_{v\in
  X_i\setminus L_i}\fut_i(v)\), giving \(|B_{i,m_i}|\le
  (\Delta-\lambda+1)(\tw+1)                      \le \twcoe{}(\tw+1)\). For a
  join bag $X_i$ with children $i_1,i_2$ (for other bags, the analysis is
  similar and simpler):
  \[
    \begin{aligned}
      \pre_i(v)=  & \pre_{i_1}(v)\cup \pre_{i_2}(v),\quad \forall v\in X_i,                                                                                                                    \\
      B_{i,1}=    & X_i\cup \bigcup_{v\in L_{i_1}\cap L_{i_2}}\pre_{i}(v) \cup \bigcup_{b\in \{1,2\}, v\in X_i\setminus L_{i_b}} \fut_{i_b}(v),                                                \\
      B_{i,2}=    & X_i\cup \bigcup_{v\in L_{i_1}\cap L_{i_2}}\pre_{i}(v)\cup\bigcup_{v\in X_i\setminus (L_{i_1}\cup L_{i_2})} \fut_{i}(v),                                                            \\
      B_{i,2k+1}= & X_i\cup \bigcup_{v\in L_i\cup \{t_{k+1},t_{k+2},\dots,t_{\ell}\}}\pre_{i}(v)\cup\bigcup_{v\in X_i\setminus L_i\setminus \{t_{k+1},t_{k+2},\dots,t_{\ell}\}} \fut_{i}(v). \\
    \end{aligned}
  \]
  For any $v\in L_i$, if it is heavy in at least one of $X_{i_1}, X_{i_2}$, then
  $|\pre_i(v)|\le \Delta-\lambda$ because of the existence of another bag as the
  switching leaf; otherwise, $|\pre_i(v)|\le 2\lambda-2$. Thus,
  \[
    \begin{aligned}
      |B_{i,2}|\le|B_{i,1}|\le & (\tw+1) + \max\{\Delta-\lambda,2\lambda-2\}|L_{i_1}\cap L_{i_2}| + (\Delta-\lambda)(\tw+1-|L_{i_1}\cap L_{i_2}|) \\
      \le                      & \twcoe{}(\tw+1).
    \end{aligned}
  \]
  Similar argument gives \(|B_{i,2k+2}|\le \twcoe{}(\tw+1)\) for $k\in [\ell]$, leading to
  \begin{fakedisplaymath}
    |B_{i,2k+1}|\le |B_{i,2k}|+\Delta-\lambda\le \twcoe{}(\tw+1)+\lfloor \textstyle \frac{2\Delta}{3} \rfloor.
  \end{fakedisplaymath}
\end{proof}

\begin{theorem}\label{thm:ctw-sqr-ub}
  For a digraph $G$ with maximum out-degree $\Delta$, we have
  \(
    \pw(G^{\gsym})\le \ctw(G)+\Delta.
  \)
  Furthermore, given a linear ordering of $G$ with cutwidth $\ctw$, a path
  decomposition of $G^{\gsym}$ of width at most $\ctw+\Delta$ can be found in
  polynomial time.
\end{theorem}
\begin{proof}
  Given a linear ordering $v_1,v_2,\ldots,v_n$ of the vertices that achieves the
  cutwidth $\ctw$, we create a path decomposition of $2n$ bags, and for each
  $i\in [n]$ (let $B_0=\emptyset$),
  \[
    B_{2i-1}=B_{2i-2}\cup N^{+}[v_i],\qquad
    B_{2i}=B_{2i-1}\setminus \{v_j\mid j\le i, \not\exists k>i \text{ s.t. } v_j\in N^{+}(v_k)\}.
  \]
  Namely, for $B_{2i}$, we remove all the heads {which are} only incident to
  $v_i$ as a tail. Now we show that this is a valid path decomposition of
  $G^{\gsym}$, and it is easy to see that every vertex and edge is covered by
  some bag. Then we can see that $v_i$ is introduced at $B_{2j-1}$ if $v_j$ is
  the leftmost vertex belonging to $N^{-}[v_i]$, and is forgotten at $B_{2k}$ if
  $v_k$ is the rightmost vertex belonging to $N^{-}[v_i]$. Therefore, it
  instantly follows that the bags containing $v_i$ are consecutive.
  
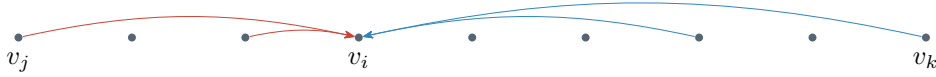
\begin{center}
    \begin{tikzpicture}[yscale=0.6]
      \foreach \i in {1,...,9} {
        \node[fill vertex] (n\i) at (1.5 * \i - 1.5, 0) {};
      }
      \node[below=2pt] at (n1) {$v_j$};
      \node[below=2pt] at (n4) {$v_i$};
      \node[below=2pt] at (n9) {$v_k$};

      \draw[diedge, mRose] (n1) to[bend left=20] (n4);
      \draw[diedge, mRose] (n3) to[bend left=20] (n4);

      \draw[diedge, mBlue] (n7) to[bend right=20] (n4);
      \draw[diedge, mBlue] (n9) to[bend right=20] (n4);
    \end{tikzpicture}
  \end{center}
  Finally, it remains to show the width of the path decomposition. Consider the
  set of edges that cross the cut after $v_{i}$, and let $H_i$ be the set of
  \textit{heads} of these edges. We claim for each $i\in [n]$, $B_{2i}=H_i$,
  from which it immediately follows that $|B_{2i}|\le \ctw$ and $|B_{2i-1}|\le
  \ctw+\Delta+1$. This claim is equivalent to saying
  \[
    H_i=H_{i-1}\cup N^{+}[v_i]\setminus \{v_j\mid j\le i, \not\exists k>i \text{ s.t. } v_j\in N^{+}(v_k)\},
  \]
  which is true since $H_i\setminus H_{i-1}$ contains the heads $v_j$ incident
  to $v_i$ that $j>i$ (marked red), and $H_{i-1}\setminus H_i$ are the heads
  whose rightmost in-neighbor is exactly $v_i$ (marked blue). Note that the
  latter set can even include $v_i$ itself if there is no in-neighbor of $v_i$
  to the right.

  
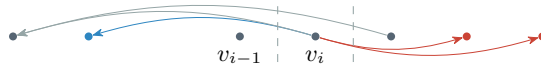
\begin{center}
    \begin{tikzpicture}[yscale=0.8]
      \node[fill vertex, label=below:{\small $v_{i-1}$}] (vim1) at (0, 0) {};
      \node[fill vertex, label=below:{\small $v_i$}] (vi) at (1,0) {};

      \node[fill vertex, mBlue] (n1) at (-2,0) {};
      \node[fill vertex] (n2) at (-3,0) {};

      \draw[diedge, mBlue] (vi) to[bend right=20] (n1);
      \draw[diedge] (vi) to[bend right=20] (n2);

      \node[fill vertex] (n3) at (2,0) {};
      \draw[diedge] (n3) to[bend right=20] (n2);

      \node[fill vertex, mRose] (n4) at (3,0) {};
      \node[fill vertex, mRose] (n5) at (4,0) {};

      \draw[diedge, mRose] (vi) to[bend right=20] (n4);
      \draw[diedge, mRose] (vi) to[bend right=20] (n5);

      \draw[edge, dashed] ($(vim1)!0.5!(vi)+(0,0.5)$) to ++(0,-1);
      \draw[edge, dashed] ($(vi)!0.5!(n3)+(0,0.5)$) to ++(0,-1);
    \end{tikzpicture}
  \end{center}
\end{proof}

\subsection{Algorithmic implications}\label{subsec:alg-implication}

\begin{theorem}\label{thm:csp-encoding} The existence of a PNE in an
  $\alpha$-strategy graphical game $\mathcal{G}=(G,\{M_v\}_{v\in V(G)})$ with
  maximum out-degree $\Delta$ can be encoded as a CSP instance
  $\psi_{\mathcal{G}}$ with primal graph $G^{\gsym}$ and alphabet size
  $\alpha$.
\end{theorem}
\begin{proof}
  For each player $v$, we create a variable $x_v$ in $\psi_{\mathcal{G}}$ whose
  domain is the strategy set of $v$, and add a constraint on the variables
  corresponding to $N^{+}[v]$ to ensure that $v$ is locally stable. Namely, we
  create a constraint whose relation is the set of all joint strategies of
  $N^{+}[v]$ that make $v$ locally stable. Now we look at the primal graph of
  $\psi_{\mathcal{G}}$. The primal graph has a vertex set $\{x_v\mid v\in
  V(G)\}$, and there is an edge between $x_u$ and $x_v$ if they appear together
  in the same constraint, and this can only happen if $u$ and $v$ are adjacent
  or they have a common in-neighbor. Therefore, the primal graph of
  $\psi_{\mathcal{G}}$ is exactly $G^{\gsym}$.
\end{proof}
Recall the \FPT{} result for CSPs parameterized by $\tw$:
\begin{theorem}[{\cite[Theorem 6 (implicit)]{gottlob_fixed-parameter_2002}}]\label{thm:csp-tw-alg}
  There is an algorithm that given a CSP instance $\psi$ with $n$ variables on
  alphabet $\Sigma$, decides whether $\psi$ is satisfiable in time
  $|\Sigma|^{\tw+1}{|\psi|}^{O(1)}$, where $\tw$ is the width of a given
  tree decomposition of the primal graph of $\psi$.
\end{theorem}
\cref{thm:csp-tw-alg,thm:csp-encoding,thm:pw-sqr-ub,thm:tw-sqr-ub,thm:ctw-sqr-ub}
together give the claimed three algorithmic results listed in \cref{thm:upper}.
Moreover, it is known that improving the algorithm in \cref{thm:csp-tw-alg} when
parameterized by \emph{pathwidth} is \textit{equivalent} to falsifying
\pwseth{}~\cite[Theorem 3.2]{lampis_primal_2025}. Therefore, if \pwseth{} is
false, then we can improve our algorithms for PNE, giving one direction of the
equivalence relation in \cref{thm:lower}.

\subsection{Tightness of the bounds}

In this section we show that the combinatorial upper bounds we have established
are tight. In particular, the constant factors $\frac{1}{2}$ and $\frac{2}{3}$
in \cref{thm:pw-sqr-ub,thm:tw-sqr-ub} cannot be improved. For this, we construct
infinite families of graphs with maximum degree and pathwidth (treewidth)
arbitrarily large such that the bounds are tight up to small additive terms.

We remark that the family we construct to show our pathwidth bound is tight is
also used, with minor modifications, as the basis of our \pwseth-based lower
bound. This makes sense: to obtain a tight lower bound, we need a family of
graphs such that $G^{\gsym}$ has pathwidth as high as we estimated, because
otherwise applying a CSP algorithm to our instances would solve the problem too
efficiently, contradicting our lower bound.
\begin{theorem}\label{thm:pw-sqr-lb}%
  For any integers $k, p\ge 2$, there exists a graph $G_{k,p}$ with maximum
  degree $\Delta=2k-1$ and pathwidth at most $p+1$ such that
  \(
    \pw(G^2_{k,p})\ge \textstyle \pwcoe{}\pw(G_{k,p})- \textstyle \pwcoe{}.
  \)
\end{theorem}
\begin{proof}
  Given integers $k,p\ge 2$, we construct a graph $G_{k,p}$ with maximum degree
  $\Delta=2k-1$ and pathwidth at most $p+1$ such that $\pw(G_{k,p}^{2})\ge kp$.

  We create a grid-like graph of $n=kp$ rows and each row contains $n$ vertices,
  and we denote them by $v_{ij}$ for $i,j\in [n]$. For the $i$-th row, we denote
  the vertices $v_{i,gk+1} (0\le g<p)$ by the \textit{pivot vertices} of this
  row. For each $v_{i,gk+1}$, we connect it to
  \(\{v_{i,gk+2},v_{i,gk+3},\dots,v_{i,(g+1)k}\}\), and also to
  \(\{v_{i+1,gk+2},\dots,v_{i+1,(g+1)k}\} \) if $i\not=n$. Finally, we connect
  $(v_{i,gk},v_{i,gk+1})$ for each $i\in [n], 0< g<p$ (see below for an example
  for $k=3, p=3$).

  
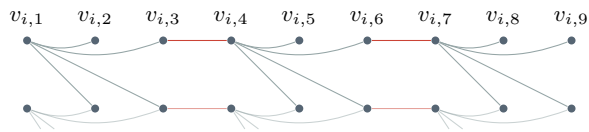
\begin{center}
    \begin{tikzpicture}[scale=0.9, label font/.style={font=\footnotesize}]
      \def\n{9}

      \foreach \j in {1,...,\n} {
        \node[fill vertex, label={[label font]above:$v_{i,\j}$}] (xi\j) at (\j, 1) {};
        \node[fill vertex] (xiplus1\j) at (\j, 0) {};
      }

      \foreach \start in {1, 4, 7} {
        \pgfmathtruncatemacro{\nextone}{\start+1}
        \pgfmathtruncatemacro{\nexttwo}{\start+2}

        \draw[edge] (xi\start) edge[bend right=20] (xi\nextone);
        \draw[edge] (xi\start) edge[bend right=20] (xi\nexttwo);

        \draw[edge] (xi\start) -- (xiplus1\nextone);
        \draw[edge] (xi\start) -- (xiplus1\nexttwo);

        \draw[edge, opacity=0.5] (xiplus1\start) edge[bend right=20] (xiplus1\nextone);
        \draw[edge, opacity=0.5] (xiplus1\start) edge[bend right=20] (xiplus1\nexttwo);

        \draw[edge, opacity=0.5] (xiplus1\start) -- ++(0.2, -0.3);
        \draw[edge, opacity=0.5] (xiplus1\start) -- ++(0.4, -0.3);
      }

      \foreach \gk in {3, 6} {
        \pgfmathtruncatemacro{\next}{\gk+1}
        \draw[edge,mRose] (xi\gk) -- (xi\next);
        \draw[edge,mRose, opacity=0.5] (xiplus1\gk) -- (xiplus1\next);
      }
    \end{tikzpicture}
  \end{center}
  Now we give a path decomposition for $G_{k,p}$ with width at most $p+1$. We
  first iterate through the first row from left to right, and create bags to
  introduce the vertices one by one. If a vertex is not pivotal, we immediately
  forget it, with the only exception that for a vertex $v_{1,gk}$ with $0<g<p$,
  we forget it after introducing $v_{1,gk+1}$ to cover the edge between them.
  After processing the first row, we stop at a bag of all the pivot vertices in
  the first row and all the other vertices in the first row are forgotten. Then
  we repeat the same process for the next rows one by one, with the only
  difference that we forget $v_{i-1,(g-1)k+1}$ whenever we introduce $v_{i,gk}$
  for $i\ge 2$ and $0<g\le p$. Then the biggest bags are the ones containing
  $v_{i,gk},v_{i,gk+1}$ for some $i\ge 2$, which have size $p+2$. A structured
  pseudocode is available in Appendix~\ref{app:pd-1}.

  \begin{toappendix}

    \appendixsection{Path decomposition of $G_{k,p}$ in the proof of
    \cref{thm:pw-sqr-lb}}\label{app:pd-1}
    \begin{algorithmic}[1] %
      \State Initialize an empty path decomposition
      \For{$j \gets 1$ \textbf{to} $n$}
        \State Append a new bag introducing $v_{1,j}$
        \If{$j \equiv 1 \pmod k$}
          \If{$j > 1$}
            \State Append a new bag forgetting $v_{1,j-1}$
          \EndIf
        \ElsIf{$j \not\equiv 0 \pmod k$ \textbf{or} $j = n$}
          \State Append a new bag forgetting $v_{1,j}$
        \EndIf
      \EndFor

      \For{$i \gets 1$ \textbf{to} $n-1$}
        \Comment{Start with pivot vertices of row $i$, others in row $i$ forgotten}
        \For{$j \gets 1$ \textbf{to} $n$}
          \State Append a new bag introducing $v_{i+1,j}$\label{line:biggest-bag}

          \If{$j \equiv 0 \pmod k$}
            \State Append a new bag forgetting $v_{i,j-k+1}$
          \EndIf

          \If{$j \equiv 1 \pmod k$}
            \If{$j > 1$}
              \State Append a new bag forgetting $v_{i+1,j-1}$
            \EndIf
          \ElsIf{$j \not\equiv 0 \pmod k$ \textbf{or} $j = n$}
            \State Append a new bag forgetting $v_{i+1,j}$
          \EndIf
        \EndFor
      \EndFor
    \end{algorithmic}
    Line~\ref{line:biggest-bag} when $j>k$ and $j\bmod k=1$ is the step where the biggest bag in this
    decomposition is created, and it is of size $p+2$, showing $\pw(G_{k,p})\le
    p+1$.
  \end{toappendix}

  Now we analyse the pathwidth of $G_{k,p}^{2}$. We argue that $G_{k,p}^{2}$
  contains an $n\times n$ grid as a subgraph, since for any $i,j$, vertices
  $v_{i,j},v_{i,j+1}$ are either connected or both adjacent to
  $v_{i,k\lfloor\frac {j-1}{k}\rfloor+1}$ (the same holds for non-pivot vertices
    $v_{i,j}$ and $v_{i+1,j}$, and for pivot vertices they are both adjacent to
  $v_{i+1,j+1}$). Thus, the monotonicity of pathwidth gives%
  \begin{fakedisplaymath}
    \pw(G_{k,p}^{2})\ge n=kp\ge \textstyle \pwcoe{}(\pw(G)-1).
  \end{fakedisplaymath}
\end{proof}

\begin{theorem}\label{thm:tw-sqr-lb}
  For any integers $k, p\ge 2$, there exists a graph $G_{k,p}$ with maximum
  degree $\Delta=3k$ and treewidth at most $p+1$ such that
  \(
    \tw(G_{k,p}^{2})\ge \textstyle \frac{2\Delta}{3}\tw(G_{k,p})- \textstyle \frac{2\Delta}{3}-1.
  \)
\end{theorem}
Similar to the previous theorem, we give a lower bound for $\tw(G_{k,p}^2)$ by
analyzing a specific subgraph. Since the family of these subgraphs appears to be
novel to the literature, we first provide its formal definition and establish
its treewidth lower bound. For any integer $n \ge 2$, let $T_n$ be the graph
obtained by identifying one endpoint of a path $P_n$ with each vertex of a
triangle $C_3$, and let $H_n$ be the Cartesian product $H_n = T_n \square P_n$
(see the diagram below for examples of $T_3,H_3$ respectively).
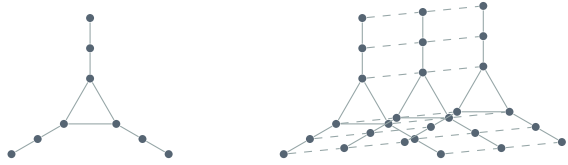
\begin{center}
  \begin{tikzpicture}[scale=0.4, every label/.append style={font=\scriptsize}]

    \begin{scope}[shift={(-1,0)}]
      \node[fill vertex] (a1) at (90:1) {};
      \node[fill vertex] (b1) at (210:1) {};
      \node[fill vertex] (c1) at (330:1) {};
      \draw[edge] (a1) -- (b1) -- (c1) -- (a1);

      \node[fill vertex] (a2) at (90:2) {};
      \node[fill vertex] (a3) at (90:3) {};
      \draw[edge] (a1) -- (a2) -- (a3);

      \node[fill vertex] (b2) at (210:2) {};
      \node[fill vertex] (b3) at (210:3) {};
      \draw[edge] (b1) -- (b2) -- (b3);

      \node[fill vertex] (c2) at (330:2) {};
      \node[fill vertex] (c3) at (330:3) {};
      \draw[edge] (c1) -- (c2) -- (c3);
    \end{scope}

    \begin{scope}[shift={(8,0)}]
      \drawHpoints{-1}{-1}{-1}{mDark};
    \end{scope}

  \end{tikzpicture}
\end{center}
We claim the following about $\tw(H_n)$, whose proof is deferred to the end of
this section.
\begin{theorem}\label{thm:bramble-hn}
  For every integer $n\ge 2$, $\tw(H_n)\ge 2n-1$.
\end{theorem}
Now we start the proof of \cref{thm:tw-sqr-lb} assuming the validity of
\cref{thm:bramble-hn}.
\begin{proof}[Proof of \cref{thm:tw-sqr-lb}]
  Given integers $k,p\ge 2$, we construct a graph $G_{k,p}$
  with maximum degree $\Delta=3k$ and treewidth at most $p+1$ such that
  $\tw(G_{k,p}^{2})\ge 2kp-1$.

  We use the grid-like construction in the proof of \cref{thm:pw-sqr-lb} as a
  gadget. Let $n=kp$, we create 3 disjoint gadgets of $n$ rows and $n$ columns,
  and denote the vertices by $\{v_{i,j}^{(h)}\mid i,j\in [n], h\in [3]\}$. Then
  we create an independent set of $p$ vertices, $r_1,r_2,\dots,r_p$, and denote
  them as \textit{central vertices}. For each $g\in [p]$, we add edges from
  $r_g$ to $v_{1,(g-1)k+j}^{(h)}$ for all $j\in [k],h\in [3]$ (see below for an
  example with $k=2,p=3$).
  
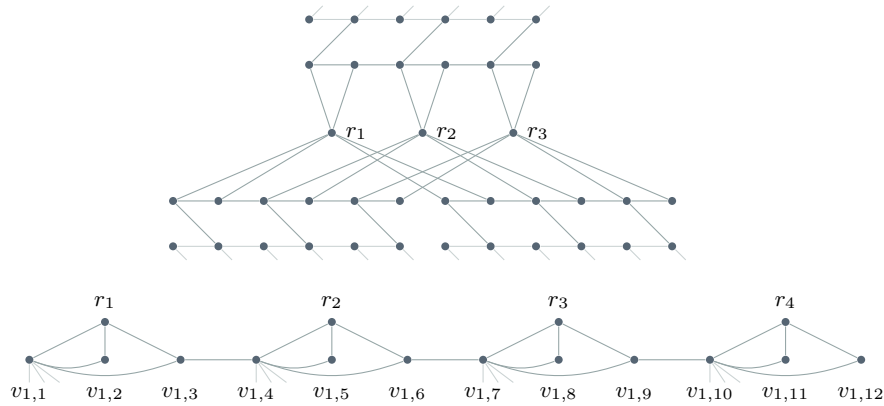
\begin{center}
    \begin{tikzpicture}
      \begin{scope}[scale=0.6, label font/.style={font=\footnotesize}]
        \newcommand{\drawgrid}[3]{
          \def\n{6}
          \foreach \j in {1,...,\n} {
            \ifnum#3>0
            \def\labelposa{above} \def\labelposb{below}
            \else
            \def\labelposa{below} \def\labelposb{above}
            \fi
            \node[fill vertex] (#1\j) at (\j, 1) {};
            \node[fill vertex] (#1plus1\j) at (\j, 0) {};
          }

          \foreach \start in {1, 3, 5} {
            \pgfmathtruncatemacro{\nextone}{\start+1}
            \draw[edge] (#1\start) -- (#1\nextone);
            \draw[edge] (#1\start) -- (#1plus1\nextone);
            \draw[edge, opacity=0.5] (#1plus1\start) -- (#1plus1\nextone);
          }

          \foreach \gk in {2, 4} {
            \pgfmathtruncatemacro{\next}{\gk+1}
            \draw[edge] (#1\gk) -- (#1\next);
            \draw[edge, opacity=0.5] (#1plus1\gk) -- (#1plus1\next);
          }

          \foreach \j in {1,...,\n} {
            \draw[edge, opacity=0.5] (#1plus1\j) -- ++(0.3,-0.3);
          }
        }

        \begin{scope}[xshift=-4.5cm, yshift=-2.5cm]
          \drawgrid{v1}{1}{1};
        \end{scope}

        \begin{scope}[xshift=1.5cm, yshift=-2.5cm]
          \drawgrid{v2}{2}{1};
        \end{scope}

        \begin{scope}[xshift=-1.5cm, yshift=2.5cm, yscale=-1]
          \drawgrid{v3}{3}{-1};
        \end{scope}

        \foreach \i in {1,...,3} {
          \node[fill vertex, label={[label font]right:$r_{\i}$}] (r\i) at (-2+2*\i,0) {};
        }

        \foreach \i in {1,...,3} {
          \pgfmathtruncatemacro{\istart}{(\i-1)*2+1}
          \pgfmathtruncatemacro{\iend}{\i*2}
          \foreach \j in {\istart,...,\iend} {
            \draw[edge] (r\i) -- (v1\j);
            \draw[edge] (r\i) -- (v2\j);
            \draw[edge] (r\i) -- (v3\j);
          }
        }
      \end{scope}
      \begin{scope}[yshift=-4cm,xshift=-5cm, scale=1, edge style/.style={line width=0.25pt}, label font/.style={font=\footnotesize}]
        \def\n{12}

        \foreach \i in {1,...,4} {
          \node[fill vertex, label={[label font]above:$r_{\i}$}] (rB\i) at (-1+3*\i,1.5) {};
        }

        \foreach \j in {1,...,\n} {
          \node[fill vertex, label={[label font,label distance=5pt]below:$v_{1,\j}$}] (xi\j) at (\j, 1) {};
        }

        \foreach \i in {1,...,4} {
          \pgfmathtruncatemacro{\istart}{(\i-1)*3+1}
          \pgfmathtruncatemacro{\iend}{\i*3}
          \foreach \j in {\istart,...,\iend} {
            \draw[edge] (rB\i) -- (xi\j);
          }
        }

        \foreach \start in {1, 4, 7, 10} {
          \pgfmathtruncatemacro{\nextone}{\start+1}
          \pgfmathtruncatemacro{\nexttwo}{\start+2}

          \draw[edge] (xi\start) edge[bend right=20] (xi\nextone);
          \draw[edge] (xi\start) edge[bend right=20] (xi\nexttwo);
          \draw[edge, opacity=0.5] (xi\start) -- ++(0, -0.3);
          \draw[edge, opacity=0.5] (xi\start) -- ++(0.2, -0.3);
          \draw[edge, opacity=0.5] (xi\start) -- ++(0.4, -0.3);
        }

        \foreach \gk in {3, 6, 9} {
          \pgfmathtruncatemacro{\next}{\gk+1}
          \draw[edge] (xi\gk) -- (xi\next);
        }
      \end{scope}
    \end{tikzpicture}
  \end{center}
  Now we give a tree decomposition of $G_{k,p}$ with width $p+1$ by first giving
  a path decomposition for the subgraph induced by $\{r_g\}_{g\in[p]}\cup
  \{v_{i,j}^{(h)}\}_{i,j\in[n]}$ for each $h\in [3]$, starting from a bag of all
  central vertices. We initialize the starting bag and process the first row as
  in the previous proof, except that $r_g$ is forgotten when $v_{1,gk}^{(h)}$ is
  introduced. Subsequent rows follow the exact same scheme, and it immediately
  follows that the maximum bag size remains $p+2$ (see
  Appendix~\ref{app:prof-path-decomp-tw-sqr-lb} for the full pseudocode).
  Finally, we can create a tree decomposition for $G_{k,p}$ by gluing the three
  path decompositions together at their common starting bag.

  \begin{toappendix}

    \appendixsection{Path decomposition for $\{r_g\}\cup \{v_{i,j}^{(h)}\}$ in
    the proof of \cref{thm:tw-sqr-lb}}\label{app:prof-path-decomp-tw-sqr-lb}
    \begin{algorithmic}[1]
      \State Initialize a bag with all central vertices
      \For{$g=1$ \textbf{to} $p$}
        \If{$g>1$}
          \State Append a new bag forgetting $r_{g-1}$
        \EndIf
        \State Append a new bag introducing $v_{1,(g-1)k+1}^{(h)}$
        \If{$g>1$}
          \State Append a new bag forgetting $v_{1,(g-1)k}^{(h)}$
        \EndIf
        \For{$j=2$ \textbf{to} $k$}
          \State Append a new bag introducing $v_{1,(g-1)k+j}^{(h)}$
          \If{$j<k$}
            \State Append a new bag forgetting $v_{1,(g-1)k+j}^{(h)}$
          \EndIf
        \EndFor
      \EndFor
      \State Append a new bag forgetting $\{r_p,v_{1,pk}^{(h)}\}$

      \For{$i \gets 1$ \textbf{to} $n-1$}
        \Comment{Start with pivot vertices of row $i$, others in row $i$ forgotten}
        \For{$j \gets 1$ \textbf{to} $n$}
          \State Append a new bag introducing $v^{(h)}_{i+1,j}$\label{line:biggest-bag-3}

          \If{$j \equiv 0 \pmod k$}
            \State Append a new bag forgetting $v^{(h)}_{i,j-k+1}$
          \EndIf

          \If{$j \equiv 1 \pmod k$}
            \If{$j > 1$}
              \State Append a new bag forgetting $v^{(h)}_{i+1,j-1}$
            \EndIf
          \ElsIf{$j \not\equiv 0 \pmod k$ \textbf{or} $j = n$}
            \State Append a new bag forgetting $v^{(h)}_{i+1,j}$
          \EndIf
        \EndFor
      \EndFor
    \end{algorithmic}
  \end{toappendix}

  Now we analyze the treewidth of $G_{k,p}^{2}$. We claim $H_n$ is a subgraph of
  $G_{k,p}^{2}$, because as previously argued, for any $h\in [3]$,
  $G_{k,p}^{2}\left[\{v_{*,*}^{(h)}\}\right]$ contains an $n\times n$ grid. And
  for any $j\in [n]$, $\{v_{1,j}^{(*)}\}\subseteq N[r_{\lfloor\frac
  {j-1}{k}\rfloor+1}]$, so $G_{k,p}^{2}\left[\{v_{1,j}^{(*)}\}\right]$ induce a
  triangle. Plugging the three grids and $n$ triangles together gives us $H_n$.
  Then the monotonicity of treewidth gives
  \begin{fakedisplaymath}
    \tw(G_{k,p}^{2})\ge 2n-1=2kp-1\ge  \textstyle \frac {2\Delta}{3}(\tw(G)-1)-1.
  \end{fakedisplaymath}
\end{proof}

Now we prove \cref{thm:bramble-hn} for completeness. We prove it by constructing
a bramble of order $2n$. We first give the formal definition of bramble, as well
as its duality relation with tree decomposition, which is the main tool we use
to prove \cref{thm:bramble-hn}.
\begin{definition}[Bramble]
  A \textit{bramble} of a graph $G$ is a family $\mathcal{B}$ of mutually
  touching connected subgraphs of $G$, i.e., for every pair $B_1, B_2 \in
  \mathcal{B}$, either they share at least one vertex, or there is an edge in
  $G$ with one endpoint in $B_1$ and the other endpoint in $B_2$. The
  \textit{order} of a bramble $\mathcal{B}$ is the size of its minimum
  hitting set.
\end{definition}
\begin{theorem}[{\cite[Theorem 1.4]{seymourGraphSearchingMinMax1993}}]
  $G$ has a bramble of order $\ge k$ if and only if $\tw(G)\ge k-1$.
\end{theorem}
\begin{center}
  \begin{tikzpicture}[scale=0.4, every label/.append style={font=\scriptsize}]
    \begin{scope}[shift={(-4,0)}, scale=0.8]
      \drawHpoints{2}{-1}{-1}{mRose};
      \node at (2,-2.5) {level $2$};
    \end{scope}
    \begin{scope}[shift={(5,0)}, scale=0.8]
      \drawHpoints{-1}{3}{-1}{mRose};
      \node at (2,-2.5) {layer $3$};
    \end{scope}
    \begin{scope}[shift={(14,0)}, scale=0.8]
      \drawHpoints{-1}{-1}{1}{mRose};
      \node at (2,-2.5) {page $1$};
    \end{scope}
  \end{tikzpicture}
\end{center}
Before we start, we first define a coordinate system for the vertices of $T_n$
and $H_n$. We denote the vertices of $T_n$ using a two-dimensional coordinate
$(j,h), j\in [n], h\in [3]$, where $j$ is the distance to the triangle plus $1$
and $h$ is the index of the path. And the vertices of $H_n$ are denoted by
$(i,j,h)$, where $i\in [n]$ is the index of the $T_n$ in the Cartesian product,
and $j,h$ are the coordinates in $T_n$. We denote the first, second and third
coordinate of a vertex by the \textit{level}, \textit{layer} and \textit{page}
of the vertex, respectively.

For $i,j,h (i,j\in [n], h \in [3])$, we define $c(i,j,h)$ to be the subgraph of
$H_n$ induced by the vertices of coordinates $(i,*,h),(i,*,h\bmod 3+1),(*,j,h)$.
Intuitively, $c(i,j,h)$ is the \textit{crossing} of the following two paths:
\begin{itemize}
  \item the intersection of pages $h,h\mod 3+1$ and level $i$ which is a
    $P_{2n}$, and
  \item the intersection of layer $j$ and page $h$ which is a $P_n$.
\end{itemize}
Now we claim the following:
\begin{proposition}\label{prop:bramble}
  The set $\mathcal{B}=\{c(i,j,h)\mid i,j\in [n], h\in [3]\}$ is a bramble of
  $H_n$.
\end{proposition}
\begin{proof}
  We prove this by showing every two elements in $\mathcal{B}$ share at least
  one vertex. Consider two arbitrary elements $c(i_1,j_1,h_1), c(i_2,j_2,h_2)$
  in $\mathcal{B}$, we have the following cases:
  \begin{itemize}
    \item if $h_1=h_2$, their crossings are on the same page, so both of them
      pass $(i_1,j_2,h_1)$ and $(i_2,j_1,h_1)$;
    \item if $h_1\neq h_2$, without loss of generality we can assume
      $h_2=h_1\mod 3+1$, then both of them pass $(i_1,j_2,h_2)$.
  \end{itemize}
  
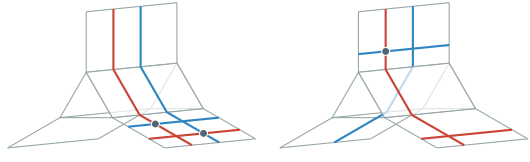
\begin{center}

    \begin{tikzpicture}[scale=0.4]
      \newcommand{\crossing}[3]{
        \coordinate (a0) at ($(v#13_1)!#2!(v#13_2)$);
        \coordinate (a1) at ($(v#11_1)!#2!(v#11_2)$);
        \node[fill vertex, draw=white] at ($(a0)!#3!(a1)$) {};
      }
      \begin{scope}
        \drawH{};
        \coverafter{};

        \draw[mRose,thick] ($(v33_1)!.3!(v33_2)$) -- ($(v31_1)!.3!(v31_2)$) -- ($(v11_1)!.3!(v11_2)$) -- ($(v13_1)!.3!(v13_2)$);
        \draw[mRose,thick] ($(v31_1)!.7!(v33_1)$) -- ($(v31_2)!.7!(v33_2)$);

        \draw[mBlue,thick] ($(v33_1)!.6!(v33_2)$) -- ($(v31_1)!.6!(v31_2)$) -- ($(v11_1)!.6!(v11_2)$) -- ($(v13_1)!.6!(v13_2)$);
        \draw[mBlue,thick] ($(v31_1)!.3!(v33_1)$) -- ($(v31_2)!.3!(v33_2)$);

        \crossing{3}{0.3}{0.7};
        \crossing{3}{0.6}{0.3};
      \end{scope}

      \begin{scope}[shift={(9,0)}]
        \drawH{};

        \draw[mBlue,thick] ($(v23_1)!.6!(v23_2)$) -- ($(v21_1)!.6!(v21_2)$) -- ($(v11_1)!.6!(v11_2)$) -- ($(v13_1)!.6!(v13_2)$);
        \draw[mBlue,thick] ($(v11_1)!.3!(v13_1)$) -- ($(v11_2)!.3!(v13_2)$);

        \coverafter{};

        \draw[mRose,thick] ($(v33_1)!.3!(v33_2)$) -- ($(v31_1)!.3!(v31_2)$) -- ($(v11_1)!.3!(v11_2)$) -- ($(v13_1)!.3!(v13_2)$);
        \draw[mRose,thick] ($(v31_1)!.7!(v33_1)$) -- ($(v31_2)!.7!(v33_2)$);

        \crossing{1}{0.3}{0.7};
      \end{scope}
    \end{tikzpicture}
  \end{center}
  Therefore, every two elements in $\mathcal{B}$ touch each other, and
  $\mathcal{B}$ is a bramble.
\end{proof}

Now we give a lower bound for the order of $\mathcal{B}$:
\begin{proposition}\label{prop:hitting-set} %
  Every hitting set of $\mathcal{B}$ has size at least $2n$.
\end{proposition}
\begin{proof}
  We show any vertex set $X$ of size $2n-1$ fails to hit at least one element
  in $\mathcal{B}$.

  We first investigate the distribution of the vertices across levels. Since
  there are $n$ levels and $|X|=2n-1$, there exists a level $i_0$ such that
  $X$ contains at most one vertex in it. Without loss of generality, we assume
  on level $i_0$ page $1$ and $2$ have no vertex in $X$. Assuming there are a
  total of $m (m\ge 0)$ vertices in $X$ that are on page $3$, there are no
  more than $2n-m-1$ vertices in $X$ on pages $1$ and $2$ in total. Now we
  look at pages $1$ and $2$. The subgraph induced by them is a grid of size
  $n\times 2n$, and we know that on level (row) $i_0$ there is no vertex in $X$.
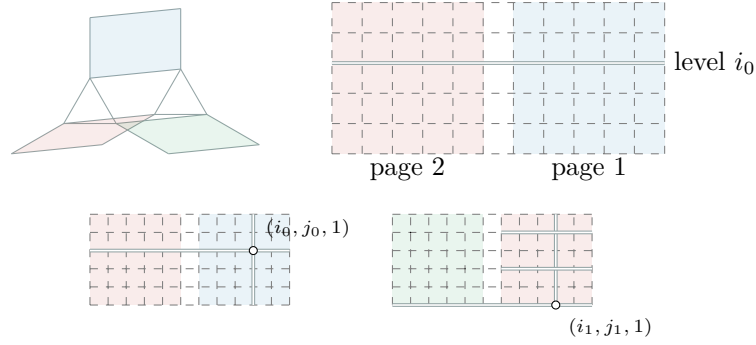
  \begin{center}
    \begin{tikzpicture}[scale=0.4, every label/.append style={font=\scriptsize}]
      \drawH{};
      \draw[draw=none, fill=mBlue, opacity=0.1] (v11_1) -- (v13_1) -- (v13_2) -- (v11_2) -- (v11_1);
      \draw[draw=none, fill=mRose, opacity=0.1] (v21_1) -- (v23_1) -- (v23_2) -- (v21_2) -- (v21_1);
      \draw[draw=none, fill=mGreen, opacity=0.1] (v31_1) -- (v33_1) -- (v33_2) -- (v31_2) -- (v31_1);

      \def\n{5}
      \pgfmathtruncatemacro{\m}{2*\n}

      \newcommand{\drawgrid}{
        \draw[draw=none, fill=mRose, opacity=0.1] (0,0) -- (\n,0) -- (\n,\n) -- (0,\n) -- (0,0);
        \draw[draw=none, fill=mBlue, opacity=0.1, shift={(\n+1,0)}] (0,0) -- (\n,0) -- (\n,\n) -- (0,\n) -- (0,0);
        \draw[dashed, gray, step=1.0cm] (0,0) grid +(\m+1,\n);
      }

      \begin{scope}[shift={(8,-1.5)}]
        \drawgrid{};
        \draw[edge, double] (0,3) -- ++ (\m+1,0);
        \node at (\n/2,-0.5) {page $2$};
        \node at (\n/2+\n+1,-0.5) {page $1$};
        \node at (\m+1,3) [right] {level $i_0$};
      \end{scope}

      \begin{scope}[shift={(0,-6.5)}, scale=0.6]
        \drawgrid{};
        \draw[edge, double] (0,3) -- ++ (\m+1,0);
        \draw[edge, double] (\n+4,0) -- ++ (0,\n);
        \node[fill vertex, fill=white, draw=black, label=above right:{$(i_0,j_0,1)$}] at ({\n+4}, 3) {};
      \end{scope}

      \begin{scope}[shift={(10,-6.5)}, scale=0.6]
        \draw[draw=none, fill=mGreen, opacity=0.1] (0,0) -- (\n,0) -- (\n,\n) -- (0,\n) -- (0,0);
        \draw[draw=none, fill=mRose, opacity=0.1, shift={(\n+1,0)}] (0,0) -- (\n,0) -- (\n,\n) -- (0,\n) -- (0,0);
        \draw[dashed, gray, step=1.0cm] (0,0) grid +(\m+1,\n);
        \draw[edge, double] (\n+1+3,0) -- ++ (0,\n);

        \draw[edge, double] (\n+1,4) -- ++ (\n,0);
        \draw[edge, double] (0,0) -- ++ (\m+1,0);
        \draw[edge, double] (\n+1,2) -- ++ (\n,0);

        \node[vanish vertex] at (4, 2) {};
        \node[vanish vertex] at (1, 4) {};
        \node[fill vertex, fill=white, draw=black, label=below right:{$(i_1,j_1,1)$}] at (\n+1+3,0) {};
      \end{scope}
    \end{tikzpicture}
  \end{center}
  If there exists a layer (column) $j_0$ on page $1$ that is not covered by
  any vertex in $X$, then $c(i_0,j_0,1)$ is not hit by $X$. Therefore, we
  suppose all layers on page $1$ are covered by $X$. Note that this also
  limits $m\le n-1$, because if $m=n$, then the total budget for page $1$ is
  $n-1$, leaving at least one layer on page $1$ uncovered. The budget for page
  $2$ is thus $n-m-1$, so there are at least $m+1$ layers and levels that are
  not covered by $X$. We arbitrarily pick an uncovered layer $j_1$, then for
  any uncovered level $i$, $c(i,j_1,2)$ has to be covered by some $(i,*,3)$.
  However, there are $m+1$ different such $i$s, and only $m$ vertices on page
  $3$, so there exists a level $i_1$ such that $c(i_1,j_1,2)$ is not covered
  by any vertex in $X$. Therefore, $X$ fails to hit all elements in
  $\mathcal{B}$, and the hitting set of $\mathcal{B}$ has size at least $2n$.
\end{proof}
By \cref{prop:bramble} and \cref{prop:hitting-set}, according to the duality
between treewidth and bramble order~\cite{seymourGraphSearchingMinMax1993}, we
have the lower bound for the treewidth of $H_n$ as in \cref{thm:bramble-hn}.

\section{Algorithmic lower bounds}\label{sec:seth_lower_bounds}

In this section we establish that our algorithms for pathwidth and cutwidth are
optimal under the \pwseth{}. Before we proceed let us give some intuition,
focusing on the case of pathwidth. Recall that we want to justify, for each
\emph{fixed} values of $\alpha,\Delta$, that parameter dependence is at least
$\left(\alpha^{\lfloor\frac{\Delta}{2}+1\rfloor}\right)^{\pw}$. Consider the
case $\Delta=3$, where our goal is $(\alpha^2)^\pw$. We will achieve this by
reducing from a CSP problem with alphabet size $\alpha$ parameterized by
pathwidth. In particular, for each variable $x$ of the initial instance and each
appearance of $x$ in a bag, we construct a player with $\alpha$ strategies, with
the goal that all these players will pick a strategy corresponding to the value
of $x$ in a satisfying assignment. It is now easy to represent each constraint;
the interesting part is how to ensure that all players representing the same
variable play the same strategy, while \emph{decreasing} the pathwidth of the
initial instance by a factor of $2$. We can easily do this if allowed to keep
the $\pw$ constant, by listing the players in an $m\times \pw$ grid (the left
diagram), and forcing the player in row $i$ to play the same as her copy in row
$i+1$ (and this is the essence of our cutwidth lower bound).
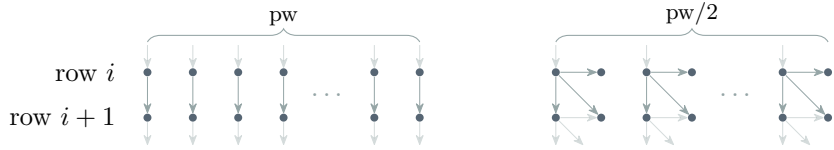
\begin{center}
  \begin{tikzpicture}[scale=0.6]
    \foreach \i in {1,...,6} {
      \ifnum\i>4
      \def\pos{\i+1}
      \else
      \def\pos{\i}
      \fi
      \node[fill vertex] (xi\i) at (\pos, 0) {};
      \node[fill vertex] (xiplus1\i) at (\pos, -1) {};
      \draw[diedge] (xi\i) -- (xiplus1\i);
      \draw[diedge,opacity=0.4] ($(xi\i)+(0,0.6)$)--(xi\i);
      \draw[diedge,opacity=0.4] (xiplus1\i)--++(0,-0.6);
    }
    \draw[decorate, edge, decoration={brace, amplitude=5pt, raise=0.4cm}]
    (xi1.center) -- (xi6.center) node[midway, above=0.5cm, align=center, font=\footnotesize] {\color{black}$\pw$};
    \node at (5,-0.5) {\color{mGrey}$\dots$};
    \node at (0.5, 0) [left] {row $i$};
    \node at (0.5, -1) [left] {row $i+1$};

    \begin{scope}[xshift=9cm]
      \foreach \i in {1,...,6} {
        \ifnum\i>4
        \def\pos{\i+1}
        \else
        \def\pos{\i}
        \fi
        \node[fill vertex] (xi\i) at (\pos, 0) {};
        \node[fill vertex] (xiplus1\i) at (\pos, -1) {};

      }
      \foreach \i in {1,3,5} {
        \pgfmathtruncatemacro{\nextone}{\i+1}
        \draw[diedge] (xi\i) -- (xi\nextone);
        \draw[diedge] (xi\i) -- (xiplus1\nextone);
        \draw[diedge] (xi\i) -- (xiplus1\i);
        \draw[diedge,opacity=0.4] (xiplus1\i) -- (xiplus1\nextone);
        \draw[diedge,opacity=0.4] ($(xi\i)+(0,0.6)$)--(xi\i);
        \draw[diedge,opacity=0.4] (xiplus1\i)--++(0,-0.6);
        \draw[diedge,opacity=0.4] (xiplus1\i)--++(0.6,-0.6);
      }

      \draw[decorate, edge, decoration={brace, amplitude=5pt, raise=0.4cm}]
      (xi1.center) -- (xi6.center) node[midway, above=0.5cm, align=center,
      font=\footnotesize] {\color{black}${\pw}/{2}$};
      \node at (5,-0.5) {\color{mGrey}$\dots$};
    \end{scope}
  \end{tikzpicture}
\end{center}
The key idea of the pathwidth construction is now to use the larger degree to
perform this information transfer more efficiently. In particular, for
$\Delta=3$, we designate every other player of a row to be a \emph{pivot} player
and make this player responsible for her own consistency as well as the
consistency of a non-pivot player. The pivot has therefore one out-neighbor in
her own row and two out-neighbors in the next row and has responsibility to
ensure that this group of $4$ players is consistent.
The pivot players in each row thus become a separator (the right diagram),
allowing us to decrease the pathwidth by the desired factor.

\begin{theorem}\label{thm:pw-lb}
  For any integers $k,\alpha\ge 2$, if there is an algorithm that given an
  $\alpha$-strategy graphical game $\mathcal{G}$ with maximum out-degree
  $\Delta=2k-1$ and pathwidth $p$, decides whether there is a PNE in time
  $\alpha^{(1-\varepsilon)k\cdot p}|\mathcal{G}|^{O(1)}$ for some
  $\varepsilon>0$, then the \pwseth{} is false.
\end{theorem}
\newcommand{\proofsection}[1]{\vspace{0.3em}\noindent\textbf{#1}.}
\begin{proof}%
  Suppose we are given a 3-CSP instance $\psi$ with domain $\Sigma$, an
  associated nice path decomposition for the primal graph of $\psi$ of width
  $\pw$ and an integer $k\ge 2$, we build a graphical game $\mathcal{G}=(G,\{M_v\})$
  satisfying the following properties:
  \begin{enumerate}
    \item $\mathcal{G}$ is $\alpha$-strategy where $\alpha=|\Sigma|$ and has
      maximum out-degree $\Delta=2k-1$; and
    \item $\mathcal{G}$ has a PNE if and only if $\psi$ is satisfiable; and
    \item $\mathcal{G}$ and a path decomposition of $G$ of width $\frac
      {\pw}{k}+7$ can be obtained in polynomial time.
  \end{enumerate}
  To see why this implies the theorem, suppose there exists an algorithm for PNE
  with runtime $\alpha^{(1-\varepsilon)k\cdot p}|\mathcal{G}|^{C}$, then we can
  use this algorithm to solve the PNE instance constructed from the CSP instance
  in time $\alpha^{(1-\varepsilon)k\cdot(\frac
  {\pw}k+7)}(n\alpha^{2k})^{C}=|\Sigma|^{(1-\varepsilon)(\pw+7k)+2kC}n^{C}$.
  When $\pw>\frac {7k+2kC}{\varepsilon}$, this is
  ${|\Sigma|}^{(1-\varepsilon')\pw}n^{O(1)}$ for some $\varepsilon'>0$, which
  contradicts \pwseth{}.

  Now we give the construction. Given a path decomposition of the primal graph
  of $\psi$, by adapting~\cite[Lemma 2.1]{lampis_primal_2025} to general CSPs
  and adding dummy constraints to $\psi$, we obtain in linear time a new nice
  path decomposition $B_1,\dots,B_\ell$ and a \textit{bijective} function $\mu$
  from the set of constraints of $\psi$ to the set of bags such that for each
  constraint $c$, $B_{\mu(c)}$ contains all the variables of $c$. Under this
  bijection, we will check each constraint with a dedicated bag.

  For all $i\in [\ell]$ and each variable $x_j$ in bag $B_i$, we create a player
  $v_{i,j}$ representing $x_j$. Assuming $B_{i}$ is used to verify a constraint
  containing $(x_{a_{i,1}},x_{a_{i,2}},x_{a_{i,3}})$, we then create players
  $u_{i,1},u_{i,2},u_{i,3}$ and add arcs from $u_{i,1}$ to
  $v_{i,a_{i,1}},v_{i,a_{i,2}},v_{i,a_{i,3}}$, from $u_{i,2},u_{i,3}$ to $u_{i,1}$ and also
  antiparallel arcs between $u_{i,2}$ and $u_{i,3}$. For two consecutive bags
  $B_i,B_{i+1}$, let $B_{i}\cap B_{i+1}=\{x_1,\ldots,x_t\}$ be the set of common
  variables they share. Then for $h=1,k+1,\dots,\lfloor\frac {t-1}{k}\rfloor
  k+1$, we create arcs from $v_{i,h}$ to \(
    \left\{v_{i,h+1},v_{i,h+2},\dots,v_{i,\min(t,h+k-1)}\right\} \cup \left\{
  v_{i+1,h},v_{i+1,h+1},\dots,v_{i+1,\min(t,h+k-1)}\right\}\). Two additional
  players $m^{(i)}_{h,1}, m^{(i)}_{h,2}$ are introduced, with antiparallel arcs
  between them and arcs from $m^{(i)}_{h,1}$ to $v_{i,h},v_{i+1,h}$. See below
  for an example of $k=3$.
  
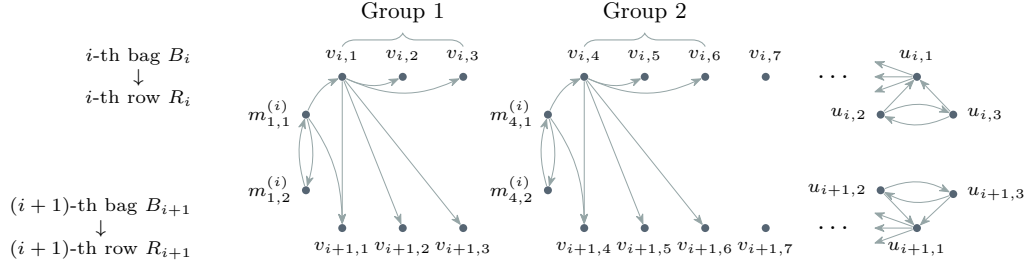
\begin{center}
    \begin{tikzpicture}[
        xscale=0.8, yscale=1,
        label font/.style={font=\scriptsize}
      ]
      \def\k{3} \def\t{7}

      \foreach \j in {1,...,\t} {
        \pgfmathsetmacro{\xcoord}{\j > 3 ? \j + 1 : \j}
        \node[fill vertex, label={[label font]above:$v_{i,\j}$}] (xi\j) at (\xcoord, 1) {};
        \node[fill vertex, label={[label font]below:$v_{i+1,\j}$}] (xiplus1\j) at (\xcoord, -1) {};
      }

      \node[left=1.8cm of xi1, font=\scriptsize, anchor=east, align=center]
      {$i$-th bag $B_i$ \\ $\downarrow$ \\ $i$-th row $R_i$};
      \node[left=1.8cm of xiplus11, font=\scriptsize, anchor=east, align=center]
      {$(i+1)$-th bag $B_{i+1}$ \\ $\downarrow$ \\ $(i+1)$-th row $R_{i+1}$};

      \node[right=0.5cm of xi\t] {$\dots$};
      \node[right=0.5cm of xiplus1\t] {$\dots$};

      \newcommand{\drawscript}[1]{
        \pgfmathsetmacro{\endTop}{\inteval{#1+\k-1}}
        \pgfmathsetmacro{\startTop}{\inteval{#1+1}}
        \foreach \j in {\startTop,...,\endTop} {
          \draw[diedge] (xi#1) to[bend left=-20] (xi\j);
        }
        \pgfmathsetmacro{\endBottom}{#1+\k-1}
        \foreach \j in {#1,...,\endBottom} {
          \draw[diedge] (xi#1) -- (xiplus1\j);
        }
        \pgfmathsetmacro{\xcoord}{#1 > 3 ? #1 + 1 : #1}
        \node[fill vertex, label={[label font]left:$m^{(i)}_{#1,1}$}] (m0) at (\xcoord-.6, 0.5) {};
        \node[fill vertex, label={[label font]left:$m^{(i)}_{#1,2}$}] (m1) at (\xcoord-.6, -0.5) {};
        \foreach \x/\y in {m0/xi#1, m0/xiplus1#1, m0/m1, m1/m0} {
          \draw[diedge] (\x) to[bend left=20] (\y);
        }
      }

      \drawscript{1}
      \drawscript{4}

      \draw[decorate, edge, decoration={brace, amplitude=5pt, raise=0.4cm}]
      (xi1.center) -- (xi3.center) node[midway, above=0.6cm, align=center, font=\footnotesize] {\color{black}Group 1};

      \draw[decorate, edge, decoration={brace, amplitude=5pt, raise=0.4cm}]
      (xi4.center) -- (xi6.center) node[midway, above=0.6cm, align=center, font=\footnotesize] {\color{black}Group 2};

      \node[fill vertex, label={[label font]above:$u_{i,1}$}] (ui1) at (10.5, 1) {};
      \node[fill vertex, label={[label font]left:$u_{i,2}$}] (ui2) at (10.5-0.6, 0.5) {};
      \node[fill vertex, label={[label font]right:$u_{i,3}$}] (ui3) at (10.5+0.6, 0.5) {};
      \draw[diedge] (ui2) -- (ui1);
      \draw[diedge] (ui3) -- (ui1);
      \draw[diedge] (ui2) edge[bend left=20] (ui3);
      \draw[diedge] (ui3) edge[bend left=20] (ui2);

      \node[fill vertex, label={[label font]below:$u_{i+1,1}$}] (uip1) at (10.5, -1) {};
      \node[fill vertex, label={[label font]left:$u_{i+1,2}$}] (uip2) at (10.5-0.6,-0.5) {};
      \node[fill vertex, label={[label font]right:$u_{i+1,3}$}] (uip3) at (10.5+0.6, -0.55) {};
      \draw[diedge] (uip2) -- (uip1);
      \draw[diedge] (uip3) -- (uip1);
      \draw[diedge] (uip2) edge[bend left=20] (uip3);
      \draw[diedge] (uip3) edge[bend left=20] (uip2);

      \foreach \x in {ui1,uip1} {
        \draw[diedge] (\x) -- ++(-0.7,0.2);
        \draw[diedge] (\x) -- ++(-0.7,0);
        \draw[diedge] (\x) -- ++(-0.7,-0.2);
      }
    \end{tikzpicture}
  \end{center}
  Now we construct the payoff matrices. Let $s_{i,j}\in \Sigma$ be the strategy
  of $v_{i,j}$. For $v_{i,h}(h\in \{1,k+1,\dots\})$, the payoff matrix of
  $v_{i,h}$ is defined in a way that it will want to play $s_{i+1,h}$ if and
  only if $s_{i,j}=s_{i+1,j}$ for all $h+1\leq j \leq \min(t,h+k-1)$. For
  $m^{(i)}_{h,1}$ and $m^{(i)}_{h,2}$, we make them a conditional matching
  pennies gadget forcing $v_{i,h}$ and $v_{i+1,h}$ to play the same strategy.
  For $u_{i,1}$, we make sure it plays $1$ if and only if the joint strategy of
  $v_{i,a_{i,*}}$ corresponds to a valid assignment for the constraint, and make
  $u_{i,2},u_{i,3}$ a conditional matching pennies gadget forcing $u_{i,1}$ to
  play $1$. From the construction we see in any PNE, copies of each variable are
  consistent, and each constraint is satisfied.%

  Finally, we give a path decomposition of width at most $\frac {\pw}k+7$. %
  The idea is similar to the ones used in the proof of \cref{thm:pw-sqr-lb},
  with slight changes listed as follows:
  \begin{itemize}
    \item before any other players in row $i+1$, introduce
      $u_{i+1,*},v_{i+1,a_{i+1,*}}$ and forget $u_{i+1,*}$; and
    \item when introducing a pivot player $v_{i+1,h}$, also introduce
      $m^{(i)}_{h,*}$ and forget them right after; and
    \item $v_{i+1,j}$ with $j\equiv 0\pmod k$ can be immediately forgotten.
  \end{itemize}
  The width can be bounded similarly. Refer to Appendix~\ref{app:pd-pw-lb} for
  details.
\end{proof}

\begin{toappendix}

  \appendixsection{Path decomposition of the game graph in the proof of \cref{thm:pw-lb}}\label{app:pd-pw-lb}

  \begin{algorithmic}[1]
      \State Initialize a bag with $u_{1,1},u_{1,2},u_{1,3},v_{1,a_{1,1}},v_{1,a_{1,2}},v_{1,a_{1,3}}$
      \State Append a new bag forgetting $u_{1,1},u_{1,2},u_{1,3}$

      \For{$j \gets 1$ \textbf{to} $n$}
        \If {$v_{1,j}$ has not been introduced}
          \State Append a new bag introducing $v_{1,j}$
        \EndIf
        \If{$j\not\equiv 1\pmod k$}
          \State Append a new bag forgetting $v_{1,j}$
        \EndIf
      \EndFor

      \For{$i \gets 1$ \textbf{to} $n-1$} \Comment{Start with pivot vertices of row $i$, others in row $i$ forgotten}
        \State Append a new bag introducing $u_{i+1,1},u_{i+1,2},u_{i+1,3},v_{i+1,a_{i+1,1}},v_{i+1,a_{i+1,2}},v_{i+1,a_{i+1,3}}$\label{line:biggest-bag-2}
        \State Append a new bag forgetting $u_{i+1,1},u_{i+1,2},u_{i+1,3}$

        \For{$j \gets 1$ \textbf{to} $n$}
          \If{$v_{i+1,j}$ has not been introduced}
            \State Append a new bag introducing $v_{i+1,j}$
          \EndIf
          \If{$j\not\equiv 1\pmod k$}
            \State Append a new bag forgetting $v_{i+1,j}$
          \Else
            \State Append a new bag introducing $m^{(i)}_{j,1},m^{(i)}_{j,2}$
            \State Append a new bag forgetting $m^{(i)}_{j,1},m^{(i)}_{j,2}$
          \EndIf
          \If{$j\equiv 0\pmod k$}
            \State Append a new bag forgetting $v_{i,j-k+1}$
          \EndIf
        \EndFor
      \EndFor
    \end{algorithmic}
  The biggest bag is created by line~\ref{line:biggest-bag-2}, containing
  $\lceil\frac {|B_i|}k\rceil$ pivot players and $6$ others. Thus, the width
  of this path decomposition is at most $\lceil \frac {|B_i|}k\rceil+5\le \frac
  {\pw+1}{k}+6<\frac {\pw}{k}+7$.
\end{toappendix}

\begin{theorem}\label{thm:ctw-lb}%
  For any integer $\alpha \ge 2$, if there is an algorithm that given an
  $\alpha$-strategy graphical game $\mathcal{G}$ with cutwidth $\ctw$, decides
  whether there is a PNE in time
  $\alpha^{(1-\varepsilon)\ctw}|\mathcal{G}|^{O(1)}$ for some $\varepsilon > 0$,
  then the \pwseth{} is false.
\end{theorem}
We refer to Appendix~\ref{app:ctw-lb} for the proof of the above theorem. These
two theorems together with the result in \cref{subsec:alg-implication} complete
the proof of equivalence for \cref{thm:lower}.%

\begin{toappendix}

  \appendixsection{Proof of \cref{thm:ctw-lb}}\label{app:ctw-lb}

  \begin{proof}
    Suppose we are given a 3-CSP instance $\psi$ with domain $\Sigma$, an
    associated nice path decomposition of the primal graph of $\psi$ with width
    $\pw$, we build a graphical game $\mathcal{G}=(G,\{M_v\})$ satisfying the
    following properties:
    \begin{enumerate}
      \item $\mathcal{G}$ is $\alpha$-strategy where $\alpha=|\Sigma|$; and
      \item $\mathcal{G}$ has a PNE if and only if $\psi$ is satisfiable; and
      \item $\mathcal{G}$ and an ordering of $V(G)$ achieving cutwidth at most
        $\pw+4$ can be constructed in polynomial time.
    \end{enumerate}
    If there exists an algorithm for PNE with runtime
    $\alpha^{(1-\varepsilon)\ctw}|\mathcal{G}|^{O(1)}$, then we can use this
    algorithm to solve CSP in time $|\Sigma|^{(1-\varepsilon)(\pw+4)}n^{O(1)}$
    using the construction above, which is
    $|\Sigma|^{(1-\varepsilon')\pw}n^{O(1)}$ for some $\varepsilon'>0$ when
    $\pw>\frac 4\varepsilon$, contradicting \pwseth{}.

    Now we give the construction. Similar to the proof of \cref{thm:pw-lb}, we
    assume we have a bijection between the constraints and the bags in the path
    decomposition, and for each variable we will also make a copy for each
    bag it appears in, and make sure all copies of the same variable play the
    same strategy.

    For each variable $x_j$ in bag $B_i$, we create a player $v_{i,j}$
    representing $x_j$, and assume $B_{i}$ is used to verify a constraint
    containing $(x_{a_{i,1}},x_{a_{i,2}},x_{a_{i,3}})$. Then create three
    players $u_{i,1},u_{i,2},u_{i,3}$ and add arcs from $u_{i,1}$ to
    $v_{i,a_{i,1}},v_{i,a_{i,2}},v_{i,a_{i,3}}$, from $u_{i,2},u_{i,3}$ to
    $u_{i,1}$ and also antiparallel arcs between $u_{i,2}$ and $u_{i,3}$. For
    two consecutive bags $B_i, B_{i+1}$ with $B_i\cap
    B_{i+1}=\{x_1,\dots,x_t\}$, we create arcs $(v_{i,j},v_{i+1,j})$ for all
    $j\in[t]$. The payoff matrices are defined in a way that $v_{i,j}$ will want
    to play the same strategy as $v_{i+1,j}$, $u_{i,1}$ will only play 1 if the
    constraint is satisfied, and $u_{i,2},u_{i,3}$ form a pair of matching
    pennies forcing $u_{i,1}$ to play 1.

    Then we have a simple layered structure where each layer corresponds to a
    bag in the path decomposition and checks the satisfaction status of the
    corresponding constraint. Therefore, we have a game graph where a PNE exists
    if and only if the original CSP instance is satisfiable. Moreover, the
    maximum out-degree in our constructed game graph is 3, so the construction
    can be done in polynomial time.

    Finally, we give an upper bound of the cutwidth of the constructed game
    graph. Suppose there are $p$ bags in total, we simply order the vertices as
    follows
    \[
      \begin{aligned}
        & v_{1,1},v_{1,2},\dots,v_{1,|B_1|},u_{1,1},u_{1,2},u_{1,3}, \\
        & v_{2,1},v_{2,2},\dots,v_{2,|B_2|},u_{2,1},u_{2,2},u_{2,3}, \\
        & \dots                                                      \\
        & v_{p,1},v_{p,2},\dots,v_{p,|B_p|},u_{p,1},u_{p,2},u_{p,3}.
      \end{aligned}
    \]
    
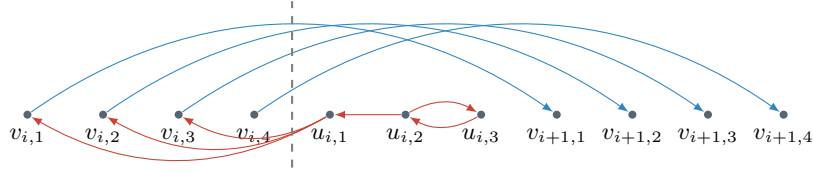
\begin{center}
      \begin{tikzpicture}[
          node distance=1cm,
          >=latex
        ]
        \node[fill vertex, label=below:{\small $v_{i,1}$}] (v_i1) {};
        \node[fill vertex, right of=v_i1, label=below:{\small $v_{i,2}$}] (v_i2) {};
        \node[fill vertex, right of=v_i2, label=below:{\small $v_{i,3}$}] (v_i3) {};
        \node[fill vertex, right of=v_i3, label=below:{\small $v_{i,4}$}] (v_i4) {};
        \node[fill vertex, right of=v_i4, label=below:{\small $u_{i,1}$}] (u_i1) {};
        \node[fill vertex, right of=u_i1, label=below:{\small $u_{i,2}$}] (u_i2) {};
        \node[fill vertex, right of=u_i2, label=below:{\small $u_{i,3}$}] (u_i3) {};

        \node[fill vertex, right of=u_i3, label=below:{\small $v_{i+1,1}$}] (v_j1) {};
        \node[fill vertex, right of=v_j1, label=below:{\small $v_{i+1,2}$}] (v_j2) {};
        \node[fill vertex, right of=v_j2, label=below:{\small $v_{i+1,3}$}] (v_j3) {};
        \node[fill vertex, right of=v_j3, label=below:{\small $v_{i+1,4}$}] (v_j4) {};

        \draw[->, mBlue] (v_i1) to[bend left=35] (v_j1);
        \draw[->, mBlue] (v_i2) to[bend left=35] (v_j2);
        \draw[->, mBlue] (v_i3) to[bend left=35] (v_j3);
        \draw[->, mBlue] (v_i4) to[bend left=35] (v_j4);

        \draw[<-, mRose] (v_i1) to[bend right=30] (u_i1);
        \draw[<-, mRose] (v_i2) to[bend right=30] (u_i1);
        \draw[<-, mRose] (v_i3) to[bend right=30] (u_i1);

        \draw[<-, mRose] (u_i2) to[bend right=30] (u_i3);
        \draw[<-, mRose] (u_i3) to[bend right=30] (u_i2);
        \draw[->, mRose] (u_i2) to (u_i1);

        \draw[dashed, thick, gray] ($(v_i4)!0.5!(u_i1)+(0,-0.7)$) -- ($(v_i4)!0.5!(u_i1)+(0,1.5)$) node[below] {};
      \end{tikzpicture}
    \end{center}
    Let $B_i\cap B_{i+1}=\{x_1,\dots,x_t\}$, then for $v_{i,j}$, the number of
    arcs in the form of $(v_{i,k},v_{i+1,k})$ that cross the cut after $v_{i,j}$
    is exactly $j$, which is at most $\pw+1$. Arcs incident to $u_{i,*}$ will
    make additional contribution at most 3. For $u_{i,*}$, the number of arcs in
    the form of $(v_{i,j},v_{i+1,j})$ that cross the cut after $u_{i,*}$ is
    exactly $|B_i|\le \pw+1$, and the arcs incident to $u_{i,*}$ will also
    contribute at most 3. Therefore, the cutwidth of the constructed ordering is
    at most $\pw+4$.
  \end{proof}
\end{toappendix}

\section{Conclusion}\label{sec:conclusion}

In this paper we established that computing a PNE is \W{1}-hard parameterized by
standard graph parameters, such as treewidth, correcting a previous claim from
the literature. We complemented this by giving optimal algorithms for this
problem for parameters pathwidth and cutwidth and an improved combinatorial
upper bound on the width of the square of a graph which is likely to find
further applications.

The most intriguing question we leave open is the gap between the problem's
complexity for pathwidth and treewidth. As mentioned, it is extremely rare to
find problems that have different complexities for these parameters.
However, the gap in our tight combinatorial relations leads us to
believe that PNE computation may be such a problem.
Finding complexity-theoretic evidence to that effect, or conversely closing the
gap by improving the treewidth-based algorithm, would therefore be very
interesting. Another interesting question would be to tighten the lower bound we
give for pathwidth to graphs where $\Delta$ is even, as at the moment our bound
is only tight for instances with odd degree.

%
\bibliography{main}

\appendix
\end{document}